\documentclass[11pt]{article}
\usepackage{ifpdf}
\usepackage{amsmath}
\usepackage{amssymb}
\usepackage{amsthm}
\usepackage[margin=1in]{geometry}
\usepackage{float}
\usepackage{graphicx}
\usepackage{comment}
\usepackage{microtype}
\usepackage{caption}
\usepackage{subfig}
\usepackage{hyperref}
\usepackage{placeins}
\usepackage{times}
\ifpdf

  \usepackage[pdftex]{epsfig}
\else

    \newcommand{\href}[2]{#2}

\fi

\theoremstyle{definition}
\newtheorem{theorem}{Theorem}[section]

\newtheorem{definition}[theorem]{Definition}

\def\compactify{\itemsep=0pt \topsep=0pt \partopsep=0pt \parsep=0pt}
\let\latexusecounter=\usecounter
\newenvironment{itemize*}
  {\def\usecounter{\compactify\latexusecounter}
   \begin{itemize}}
  {\end{itemize}\let\usecounter=\latexusecounter}
\newenvironment{enumerate*}
  {\def\usecounter{\compactify\latexusecounter}
   \begin{enumerate}}
  {\end{enumerate}\let\usecounter=\latexusecounter}
\newenvironment{description*}
  {\begin{description}\compactify}
  {\end{description}}

\begin{document}
% Author macros::begin %%%%%%%%%%%%%%%%%%%%%%%%%%%%%%%%%%%%%%%%%%%%%%%%
\title{Universal Shape Replicators via Self-Assembly with Attractive and Repulsive Forces\footnote{Research supported in part by National Science Foundation Grants CCF-1117672, CCF-1555626, and CCF-1422152.} \textsuperscript{,}\footnote{Research supported in part by National Science Foundation Grants EFRI-1240383 and CCF-1138967.} \textsuperscript{,}\footnote{6 of the pages in this paper are diagrams so do not count toward page limit.}}
\author{
Cameron Chalk\footnotemark[1]
\and
Erik D. Demaine\footnotemark[2]
\and
Martin L. Demaine\footnotemark[2]
\and
Eric Martinez\footnotemark[1]
\and
Robert Schweller\footnotemark[1]
\and
Luis Vega\footnotemark[1]
\and
Tim Wylie\footnotemark[1]
}

%mandatory. First: Use abbreviated first/middle names. Second (only in severe cases): Use first author plus 'et. al.'

%\Copyright{Cameron Chalk, Erik D. Demaine, Martin L. Demaine, Eric Martinez, Robert Schweller, Luis Vega, Tim Wylie} %mandatory, please use full first names. LIPIcs license is "CC-BY";  http://creativecommons.org/licenses/by/3.0/
%\subjclass{F.1.1. Models of Computation}
%\keywords{Tile self-assembly, 2HAM, aTAM, DNA computing, biocomputing}
% Author macros::end %%%%%%%%%%%%%%%%%%%%%%%%%%%%%%%%%%%%%%%%%%%%%%%%%
\date{}
%get rid of page number on the bottom
\clearpage\maketitle
\thispagestyle{empty}

\vspace*{-.5cm}
\begin{center}
$^*$Department of Computer Science, University of Texas - Rio Grande Valley, Edinburg, TX 78539, USA, \url{{cameron.chalk01,eric.m.martinez02,robert.schweller,luis.a.vega01,timothy.wylie}@utrgv.edu}

\vspace*{.2cm}
$^\dagger$MIT Computer Science and Artificial Intelligence Laboratory, 32 Vassar St., Cambridge, MA 02139, USA, \url{{edemaine,mdemaine}@mit.edu}
\end{center}

\begin{abstract}
We show how to design a universal \emph{shape replicator} in a self-assembly
system with both attractive and repulsive forces.
More precisely, we show that there is a universal set of constant-size objects
that, when added to \emph{any unknown} hole-free polyomino shape, produces an
unbounded number of copies of that shape (plus constant-size garbage objects).
The constant-size objects can be easily constructed from a constant number of
individual tile types using a constant number of preprocessing self-assembly
steps.
Our construction uses the well-studied 2-Handed Assembly Model (2HAM) of
tile self-assembly, in the simple model where glues interact only with
identical glues, allowing glue strengths that are either positive (attractive)
or negative (repulsive), and constant temperature (required glue strength for
parts to hold together). We also require that the given shape has specified
glue types on its surface, and that the feature size (smallest distance between
nonincident edges) is bounded below by a constant.
Shape replication necessarily requires a self-assembly model where parts can
both attach and detach, and this construction is the first to do so using the
natural model of negative/repulsive glues (also studied
before for other problems such as fuel-efficient computation); previous
replication constructions require more powerful global operations such as an
``enzyme'' that destroys a subset of the tile types.
\end{abstract}

\setcounter{page}0
\newpage

\section{Introduction}\label{sec:introduction}
%--- Out of date: \input{introduction/table.tex}

How can we harness the power of self-assembly to build a system that
first senses the shape of a given unknown object, then builds copies of
that shape, like the nanoscale equivalent of a photocopier?
Although replicators are traditionally the subject of science fiction
\cite{MemoryAlpha}, biological reproduction illustrates that this is possible
(at least approximately) in the real world, and recent biological engineering
shows that information can be replicated using
DNA crystal growth and scission \cite{SCHUL2012}.
In this paper, we investigate more complex replication of
\emph{geometric shape}, using biologically realistic models
(taken to an extreme admittedly not yet practical).
%To answer this question,
To do so, we need a sufficiently flexible algorithmic model
of self-assembly, and a replication ``algorithm'' defined by particles
that interact with each other and the given object.

%\paragraph{Self-assembly models.}
The standard algorithmic abstraction of self-assembly is to model the
self-assembling particles as \emph{Wang tiles}, that is, unit squares with a specified ``glue''
on each side, which can translate but not rotate; each glue has a nonnegative
integer \emph{strength}.
In the \emph{2-Handed Assembly Model} (2HAM), two assemblies (eventually)
join together if they can be translated so as to match up glues of total
strength at least $\tau$, the \emph{temperature} of the system.
The resulting \emph{terminal} assemblies are those that do not join into
any other assemblies.

%We can attempt to define
We can define the \emph{replication problem} in this model:
given an unknown initial assembly of tiles, design a collection of tiles
or small assemblies to add to the self-assembly system such that the resulting
terminal assemblies consist of copies of the given shape (plus possibly
some small ``trash'' assemblies).  However, this goal is impossible to achieve
in just the 2HAM model, or any model with just a mechanism to join assemblies
together: to sense a shape without destroying it, we need to be also able to
split assemblies back apart.

The first extension of the 2HAM shown to enable a solution to the
replication problem is the \emph{RNAse enzyme} staged assembly model
\cite{Abel:2010:SRT:1873601.1873686}.  In this model, tiles can be of two
types (DNA or RNA), the given shape is all one type (DNA), and there is an
operation that destroys all tiles of the other type (RNA).
Abel et al.~\cite{Abel:2010:SRT:1873601.1873686} show that this model enables
making a desired number $k$ of copies of a given unknown hole-free shape or,
through a complicated construction, infinitely many copies of the shape
using just a constant number of tiles and stages.
This result requires that the shape has a \emph{feature size} (minimum distance
between two nonincident edges) of $\Omega(\lg n)$ where $n$ is the (unknown)
number of tiles in the shape.
The model is unsatisfying, however, in the way that it requires multiple stages
and a global operation that modifies all tiles.

Another extension of the 2HAM shown to enable a partial solution to the
replication problem is the \emph{Signal Tile Assembly Model} (STAM)
\cite{keenan14exponential}.  This powerful model allows tiles to change
their glues and trigger assembly/disassembly events
when tiles attach to other tiles.  Keenan et al.~\cite{keenan14exponential}
show that this model enables replicating a given unknown \emph{pattern} of
tiles in a rectangular shape, but not an arbitrary shape.
Hendricks et al.~\cite{Hendricks2015} show that this model enables infinitely
replicating a given unknown hole-free shape of feature size at least~$2$.
The model is unsatisfying, however, in the way that it allows tiles to
have arbitrarily complex behaviors.
(Recent results show how to simulate part of STAM using 2HAM (in 3D)
\cite{Fochtman2015}, but this simulation necessarily cannot simulate
the necessary aspect of breaking assemblies apart.)

In this paper, we study a simple extension of 2HAM to allow glues of
negative strength, that is, repulsive forces in addition to standard
attractive forces.  This extension is practical, as biology implements both
types of forces \cite{Rothemund01022000}, and well-studied theoretically:
negative glues have already been shown to enable fuel-efficient
computation \cite{SS2013FEC}, space-efficient computation \cite{Doty2013},
and computation even at temperature $\tau=1$ \cite{PRS2016RMN,rgTAM}.  The complexity of combinatorial optimization problems with negative glues has also been studied~\cite{REIF20111592}.  We show that shape replication is possible in this model:
adding a fixed constant number of constant-size assemblies
to a given unknown hole-free shape results in terminal assemblies of
infinitely many copies of that shape, plus constant-size trash assemblies.

%Just a paragraph taken from the conclusion, replace with something better.
%We study the problem of shape replication in the 2-handed tile assembly model and provide a universal replication system for all genus-0 shapes with at least a constant minimum feature size.  Shape replication has been studied in more powerful self-assembly models such as the staged self-assembly model and the signal tile model.  However, our result constitutes the first example of general shape replication in a passive model of self-assembly where no outside experimenter intervention is required, and the system monomers are state-less, static pieces that interact based purely on the attraction and repulsion of surface chemistry.

%\input{introduction/stuffToCite}

%%%% Definitions Section %%%%
\section{Definitions and Model}\label{sec:definitions}

\begin{figure*}
	\begin{center}
	\includegraphics[width=.5\textwidth]{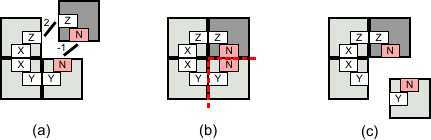}
	\caption{  A simple example of attachment and detachment events and the notation for our constructions. On each tile, the glue label is presented. Red (shaded) labels represent negative glues, and the relevant glue strengths for the tiles can be found in the captions. For caption brevity, for a glue type $X$ we denote $str(X)$ simply as $X$ (e.g. $X + Y \implies str(X) + str(Y)$). In this temperature one, $(\tau=1) $ example, $X = 2$, $Y=1$, $Z=2$, and $N=-1$. (a) The three tile assembly on the left attaches with the single tile with strength $Z + N = 2  - 1 = \tau$ resulting in the $2\times2$ assembly shown in (b). However, this $2\times2$ assembly is unstable along the cut shown by the dotted line, since $Y + N = 1 - 1 < \tau$. Then the assembly is breakable into the assemblies shown in (c).
\label{fig:modelex} }
%This figure introduces our notation for our constructions, as well as a simple example of combination and detachment events.  We denote positive glues with vertical and horizontal lines protruding from tile edges along with numbers and labels denoting glue strength and type.  Negative strength glues are labeled with slanted lines.  Finally, bonded glues between adjacent tiles are depicted with a number denoting the bond strength inscribed within a box between the bonded tile edges.  (a) The 3-tile assembly on the left and the singleton gray tile are combinable into a $2\times 2$ square shown in (b) as the cut between the 2 assemblies yields strength 1.  Note, however, that the producible $2\times 2$ square in (b) is not stable and is breakable into the assemblies shown in (c).

	\end{center}
\end{figure*}

Here we define the two-handed tile self-assembly model (2HAM) with both negative and positive strength glue types, as well as formulate the problem of designing a tile assembly system that replicates any input shape.  The 2HAM model was first presented in~\cite{2HABTO}.  The 2HAM is a variant of the abstract tile self-assembly model (aTAM) first presented in~\cite{Roth01} in which there is no seed tile, and large assemblies may combine together.

\subsection{Tile Self-Assembly Model}
\textbf{Tiles and Assemblies.} A tile is an axis-aligned unit square centered at a point in $\mathbb{Z}^2$, where each edge is labeled by a \emph{glue} selected from a glue set $\Pi$.
A \emph{strength function} ${\rm str} : \Pi \rightarrow \mathbb{Z}$ denotes the \emph{strength} of each glue.
Two tiles that are equal up to translation have the same \emph{type}.  A \emph{positioned shape} is any subset of $\mathbb{Z}^2$.  A \emph{positioned assembly} is a set of tiles at unique coordinates in $\mathbb{Z}^2$, and the positioned shape of a positioned assembly $A$ is the set of coordinates of those tiles.

For a given positioned assembly $\Upsilon$, define the \emph{bond graph} $G_\Upsilon$ to be the weighted grid graph in which each element of $\Upsilon$ is a vertex and the weight of an edge between tiles is the strength of the matching coincident glues or~0.\footnote{Note that only matching glues of the same type contribute a non-zero weight, whereas non-equal glues always contribute zero weight to the bond graph.  Relaxing this restriction has been considered as well~\cite{AGKS05g}.}  A positioned assembly $C$ is said to be \emph{$\tau$-stable} for positive integer $\tau$ provided the bond graph $G_C$ has min-cut at least $\tau$, and $C$ is said to be \emph{connected} if every pair of vertices in $G_C$ has a connecting path using only positive strength edges.

For a positioned assembly $A$ and integer vector $\vec{v} = (v_1, v_2)$, let $A_{\vec{v}}$ denote the assembly obtained by translating each tile in $A$ by vector $\vec{v}$.  An \emph{assembly} is a set of all translations $A_{\vec{v}}$ of a positioned assembly $A$.  An assembly is $\tau$-stable if and only if its positioned elements are $\tau$-stable.  An assembly is \emph{connected} if its positioned elements are connected.  A \emph{shape} is the set of all integer translations for some subset of $\mathbb{Z}^2$, and the shape of an assembly $A$ is defined to be the union of the positioned shapes of all positioned assemblies in $A$.  The \emph{size} of either an assembly or shape $X$, denoted as $|X|$, refers to the number of elements of any positioned element of $X$.

\textbf{Breakable Assemblies.}  An assembly is \emph{$\tau$-breakable} if it can be cut into two pieces along a cut whose strength sums to less than $\tau$.  Formally, an assembly $C$ is \emph{breakable} into assemblies $A$ and $B$ if $A$ and $B$ are connected, and the bond graph $G_{C'}$ for some assembly $C' \in C$ has a cut $(A',B')$ for $A'\in A$ and $B'\in B$ of strength less than $\tau$. We call $A$ and $B$ a pair of \emph{pieces} of the breakable assembly $C$.

\textbf{Combinable Assemblies.}
Two assemblies are \emph{$\tau$-combinable} provided they may attach along a border whose strength sums to at least $\tau$. Formally, two assemblies $A$ and $B$ are \emph{$\tau$-combinable} into an assembly $C$ provided $G_{C'}$ for any $C'\in C$ has a cut $(A',B')$ of strength at least $\tau$ for some $A'\in A$ and $B' \in B$.  We call $C$ a \emph{combination} of $A$ and $B$.

Note that $A$ and $B$ may be combinable into an assembly that is not stable.  This is a key property that is leveraged throughout our constructions.  See Figure~\ref{fig:modelex} for an example.

\textbf{States, Tile Systems, and Assembly Sequences.}
A \emph{state} is a multiset of assemblies whose counts are in $Z^+ \bigcup \infty$.  We use notation $S(x)$ to denote the multiplicity of an assembly $x$ in a state $S$.  A state $S_1$ \emph{transitions} to a state $S_2$ at temperature $\tau$, and is written as $S_1 \rightarrow^{\tau}_1 S_2$, if $S_2$ is obtained from $S_1$ by either replacing a pair of combinable assemblies in $S_1$ with their combination assembly, or by replacing a breakable assembly in $S_1$ by its pieces.  We simplify this to $S_1 \rightarrow_1 S_2$ when $\tau$ is clear from context.  We write $S \rightarrow^{\tau} S'$ to denote the transitive closure of $\rightarrow^{\tau}_1$, i.e., $S \rightarrow^{\tau} S'$ means that $S \rightarrow^{\tau}_1 S_1 \rightarrow^{\tau}_1 S_2 \rightarrow^{\tau}_1 \ldots S_k \rightarrow^{\tau}_1 S'$ for some sequence of states $\langle S_1, \ldots S_k \rangle$.

A \emph{tile system} is an ordered tuple $\Gamma = (\sigma, \tau )$  where $\sigma$ is a state called the \emph{initial system state} and $\tau$ is a positive integer called the \emph{temperature}.  We refer to any sequence of states $\langle \sigma, S_1, S_2, \ldots, S_k \rangle$ such that $\sigma \rightarrow^{\tau}_1 S_1$, and $S_i \rightarrow^{\tau}_1 S_{i+1}$ for $i=\{1\dots k-1\}$, as a \emph{valid assembly sequence} for $\Gamma$.  We denote any state $P$ at the end of of a valid assembly sequence (equivalently, $\sigma \rightarrow^{\tau} P$) as a \emph{producible state} of $\Gamma$, and any assembly contained in a producible state is said to be a \emph{producible assembly}.  A producible state $T$ is said to be {terminal} if $T$ does not transition at temperature $\tau$ to any other state.  A producible assembly is \emph{terminal} if it is not breakable and not combinable with any other producible assembly.

\subsection{Shape Replication Systems}

A system $\Gamma = (\sigma, \tau)$ is a \emph{universal shape replicator} for a class of shapes if for any shape $X$ in the class, there exists an assembly $\Upsilon$ of shape $X$ such that system $\Gamma' = (\sigma \bigcup \Upsilon, \tau)$ (i.e., $\Gamma$ with a single copy of $\Upsilon$ added to the initial state) produces an unbounded number of assemblies of shape $X$, and essentially nothing else.  We formalize this in the following definition, and then add some discussion of additional desirable properties a replicator might have.

\begin{definition}\label{def} [Universal Shape Replicator]
A system $\Gamma = (\sigma, \tau)$ is a \emph{universal shape replicator} for a class of shapes $U$ if for any shape $X \in U$, there exists an assembly $\Upsilon$ of shape $X$ such that system $\Gamma' = (\sigma \bigcup \Upsilon, \tau)$ has the following properties:
\begin{itemize}
    \setlength{\itemsep}{1pt}
    \setlength{\parskip}{0pt}
    \setlength{\parsep}{0pt}
    \item For any positive integer $n$ and producible state $S$, there exists a producible state $S'$ containing at least $n$ terminal assemblies of shape $X$ such that $S \rightarrow^{\tau} S'$.
    \item All terminal assemblies of super-constant size have shape $X$.
\end{itemize}
\end{definition}

In addition to the above replicator properties, there are some additional desirable properties a universal replicator might have.  For example, the producible assemblies of the system should be limited in size as much as possible ($O(|X|)$ in the best possible case).  Additionally, it may be desirable to place substantial limitations on the complexity of the initial input assembly $\Upsilon$, e.g., require its surface to expose essentially a single type of glue for each edge orientation.  Our universal replicator achieves the following \emph{bonus} constraints.

\begin{definition} [Bonus properties!]
A universal replicator is said to be \emph{sleek} if it has the following properties:
    \begin{itemize}
        \setlength{\itemsep}{1pt}
        \setlength{\parskip}{0pt}
        \setlength{\parsep}{0pt}
        \item (bonus!) All producible assemblies have size $O(|X|)$.
        \item (bonus!) The input assembly $\Upsilon$ is rather simple:  infinite internal bonds, and a single glue type for each edge orientation (e.g. \emph{North}, \emph{East}, \emph{South}, and \emph{West}) for exposed surface glues, with at most $O(1)$ \emph{special} tiles violating this convention.
    \end{itemize}
\end{definition}

The class of shapes which can be replicated by the system described in this paper are hole-free polyominoes with \emph{feature size} of at least $9$. The feature size of a shape is defined as follows (we use the same definition as~\cite{Abel:2010:SRT:1873601.1873686}): for two points $a,b$ in the shape, let $d(a,b) = \max(|a_x - b_x|, |a_y - b_y|)$. Then the feature size is the minimum $d(a,b)$ such that $a,b$ are on two non-adjacent edges of the shape. The feature size ensures that the replication gadgets ($\mathcal{O}(1)$ sized assemblies used to replicate the shape) function properly (e.g. the gadgets can encompass the shape without ``bumping into'' each other). Formally, we prove the following:
\newtheorem*{thm:main}{Theorem~\ref{thm:main}}
\begin{thm:main}
There exists a sleek universal shape replicator $\Gamma = (\sigma, 10)$ for genus-$0$ (hole-free) shapes with feature size of at least $9$.
\end{thm:main}

\section{Overview of Replication Process}\label{sec:highLevelShapeReplication}
% -*- root: ../main.tex -*-
\begin{figure}
  \centering
  \subfloat[Input Shape, $\Upsilon$ \label{fig-outer:a} ]{\includegraphics[width=0.25\textwidth]{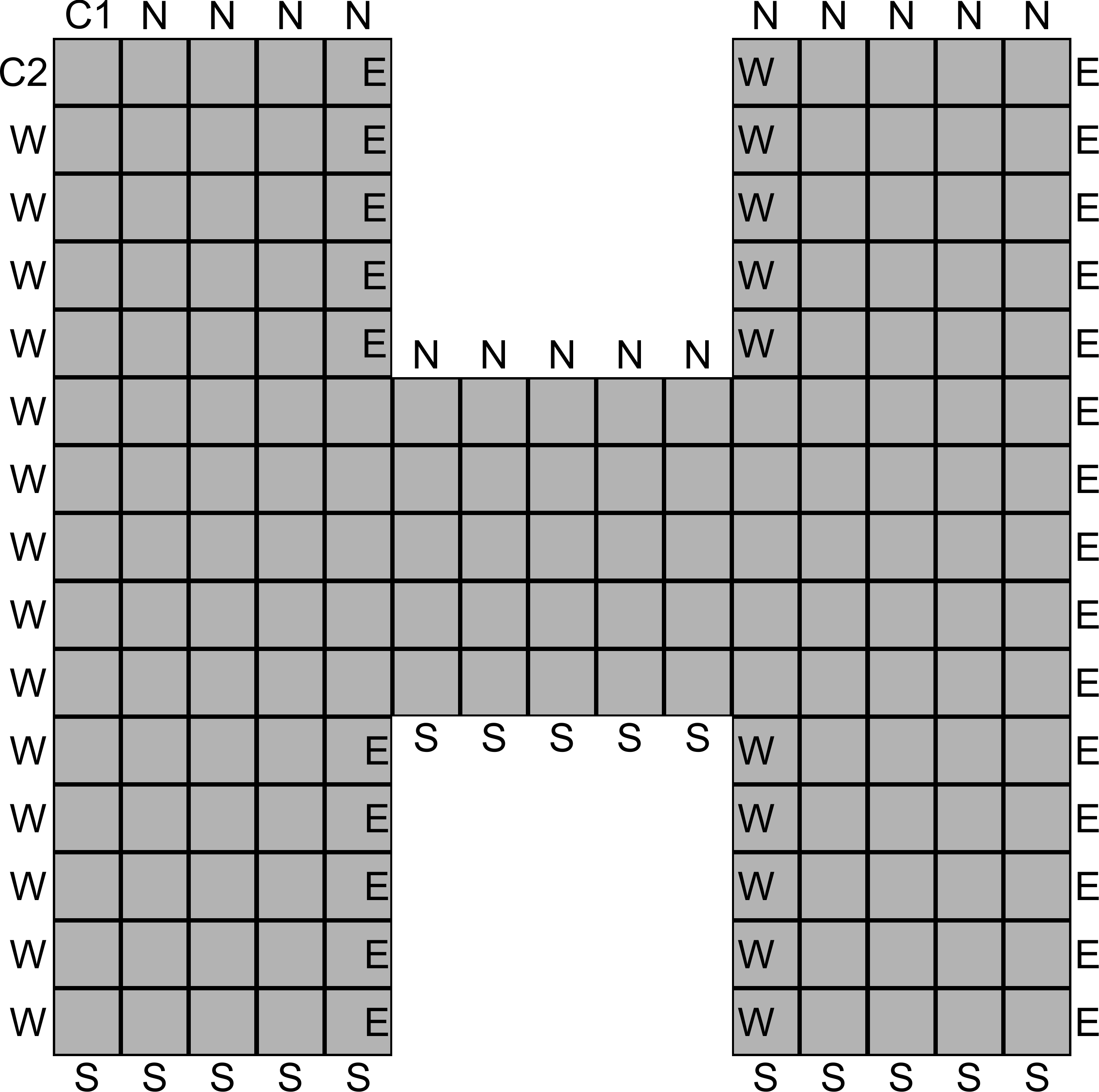}}
  \qquad %
  \subfloat[Start of process \label{fig-outer:b} ]{\includegraphics[width=0.25\textwidth]{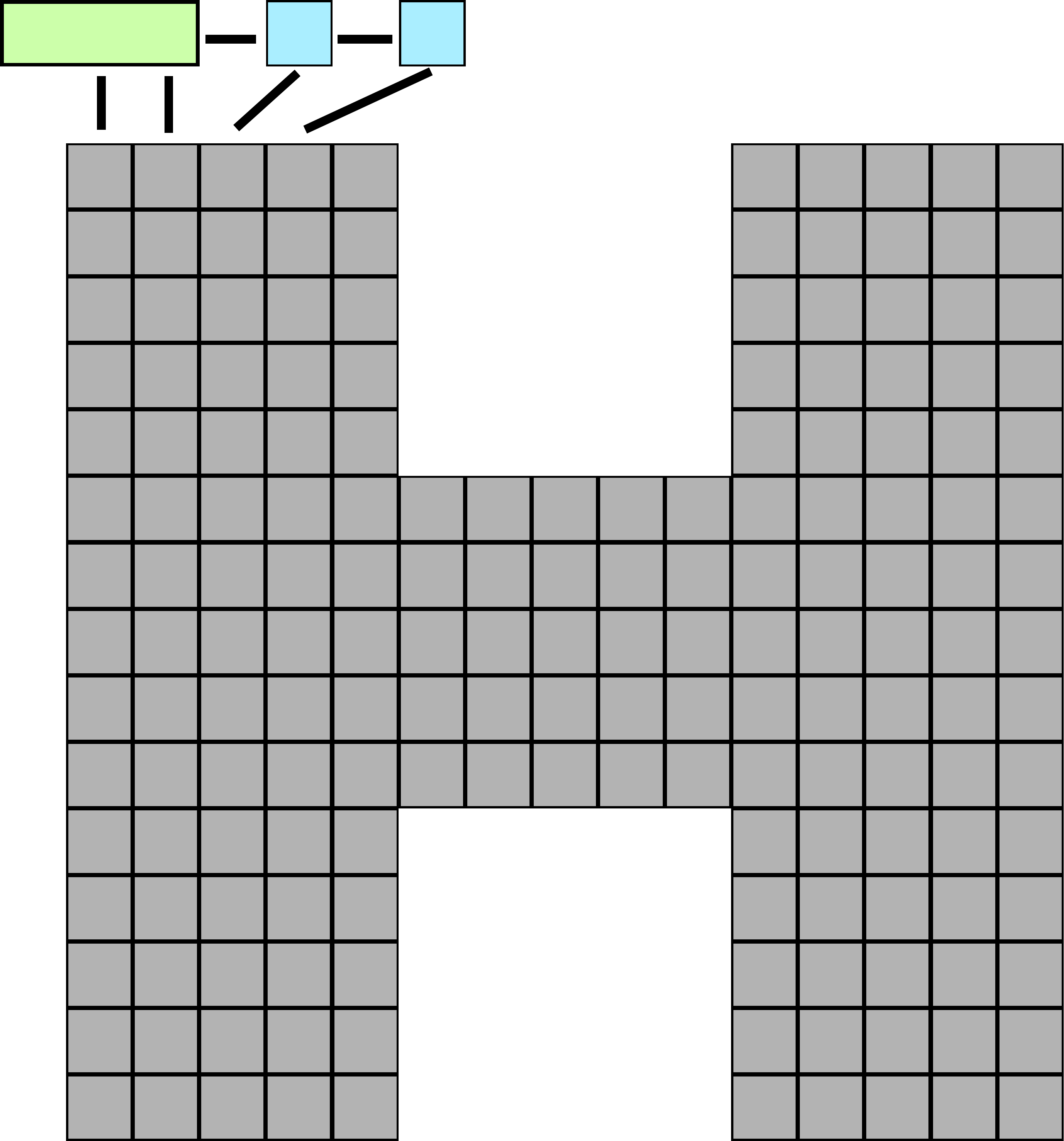}}
  \qquad %
  \subfloat[MOLD{[$\Upsilon$]} \label{fig-outer:c}]{\includegraphics[width=0.25\textwidth]{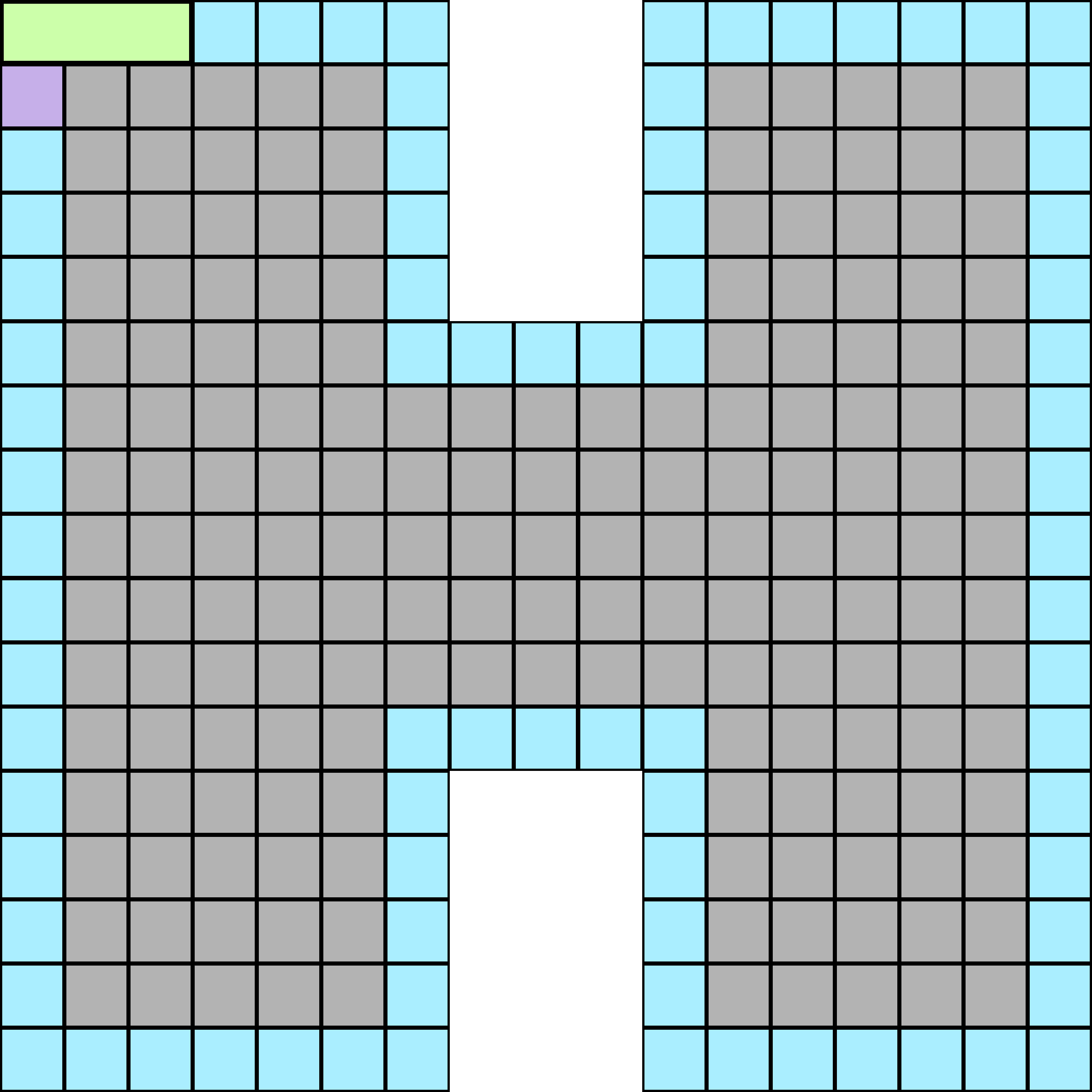}}
  \qquad %
  \subfloat[Drilling \label{fig-outer:d} ]{\includegraphics[width=0.25\textwidth]{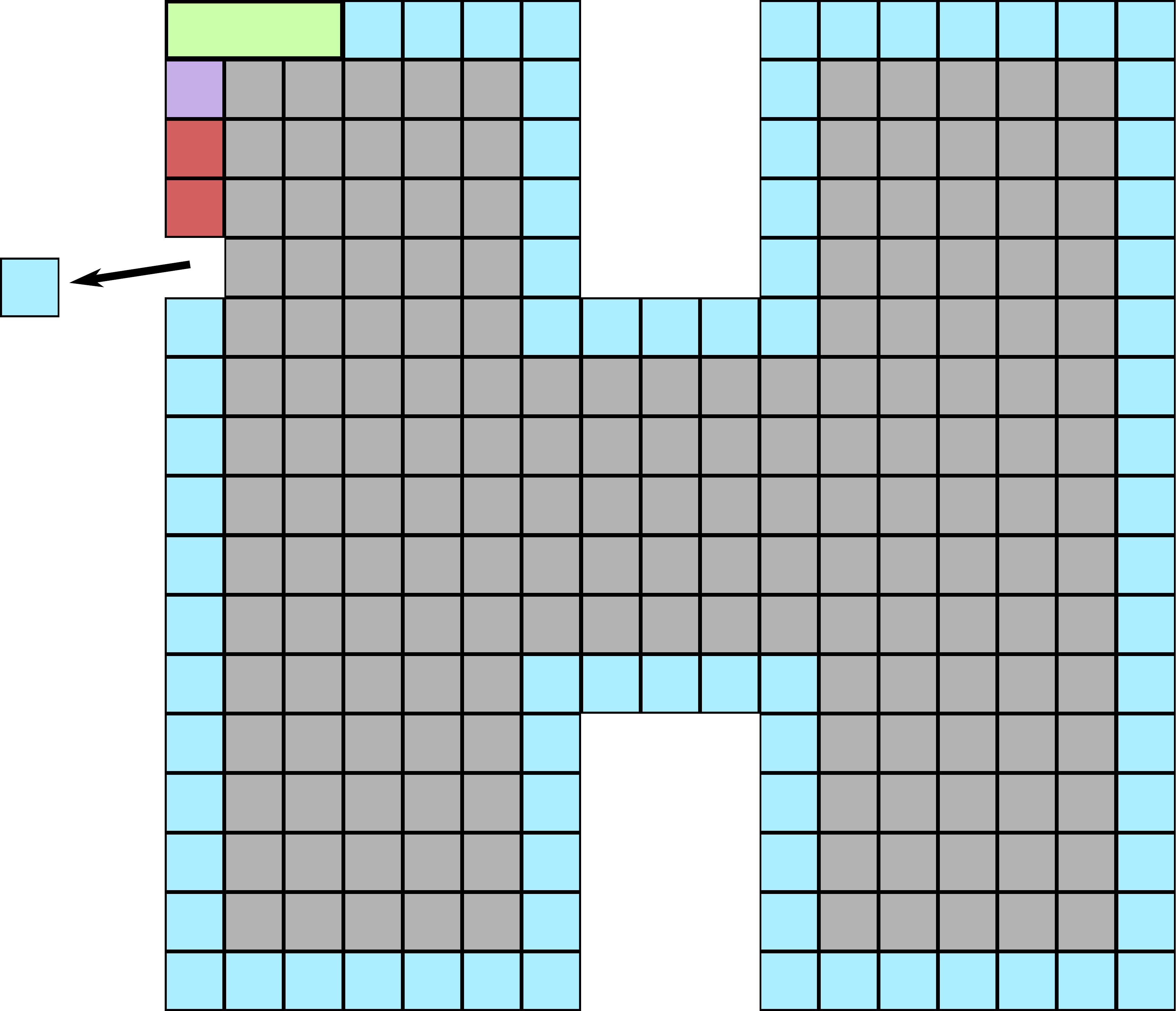}}
  \qquad %
  \subfloat[FRAME{[$\Upsilon$]} \label{fig-outer:e} ]{\includegraphics[width=0.25\textwidth]{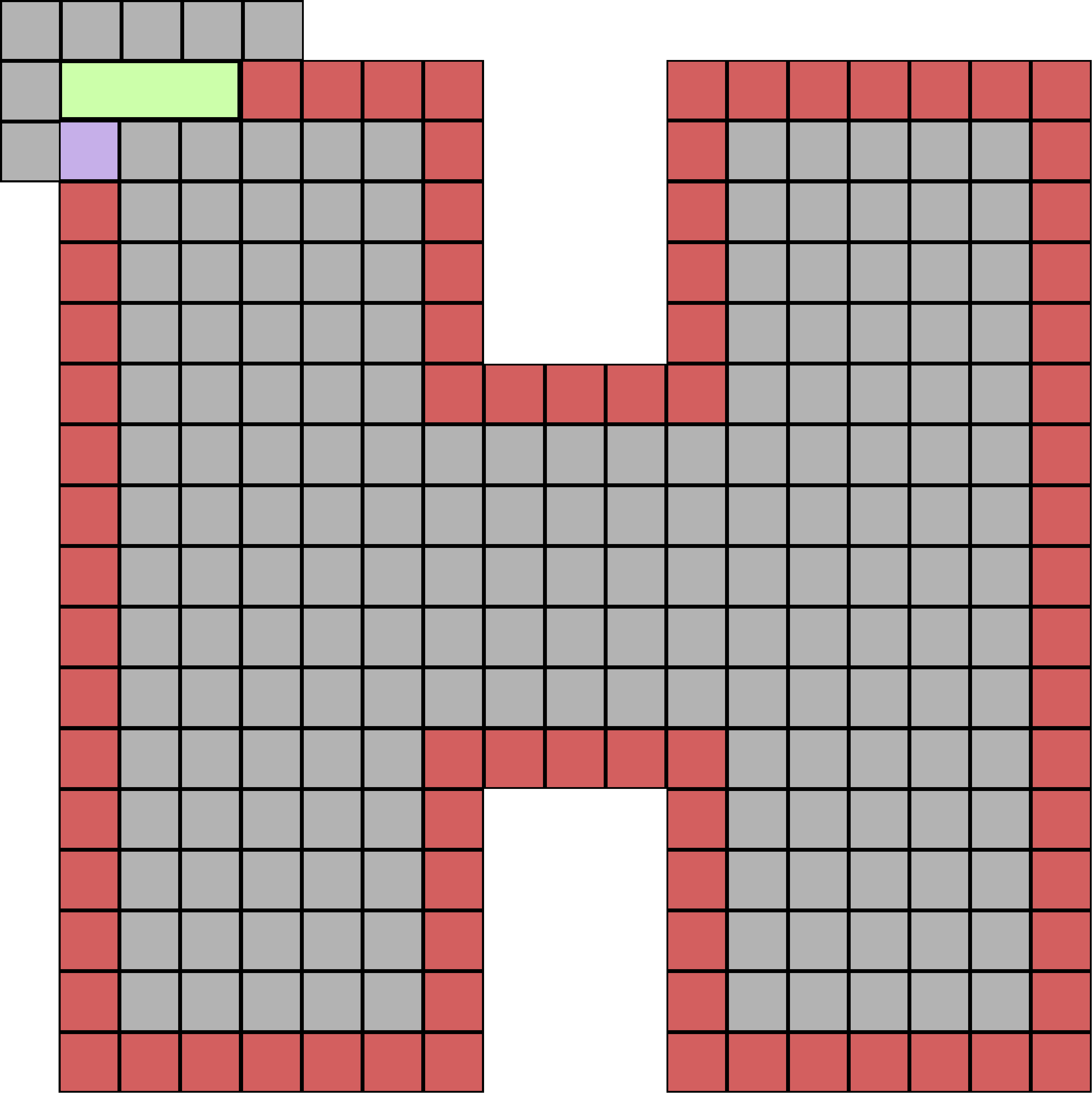}}
  \qquad %
  \subfloat[FRAME \label{fig-outer:f} ]{\includegraphics[width=0.25\textwidth]{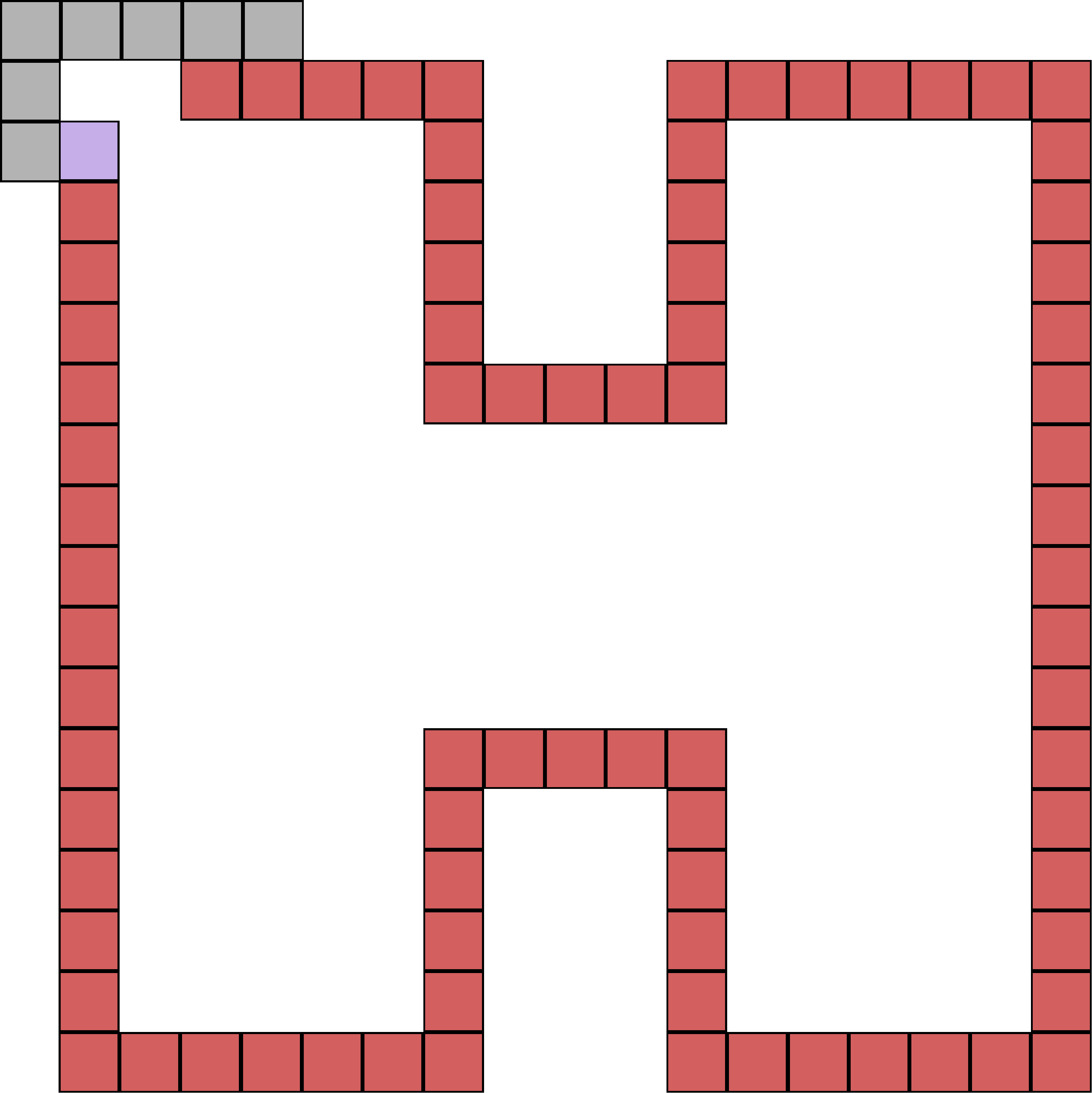}}
  \caption{High level process for making the FRAME from input shape $\Upsilon$.}
  \label{fig-outer:mold}
\end{figure}

\begin{figure}
  \centering
  \subfloat[Start of inner process \label{fig-inner:a} ]{\includegraphics[width=0.25\textwidth]{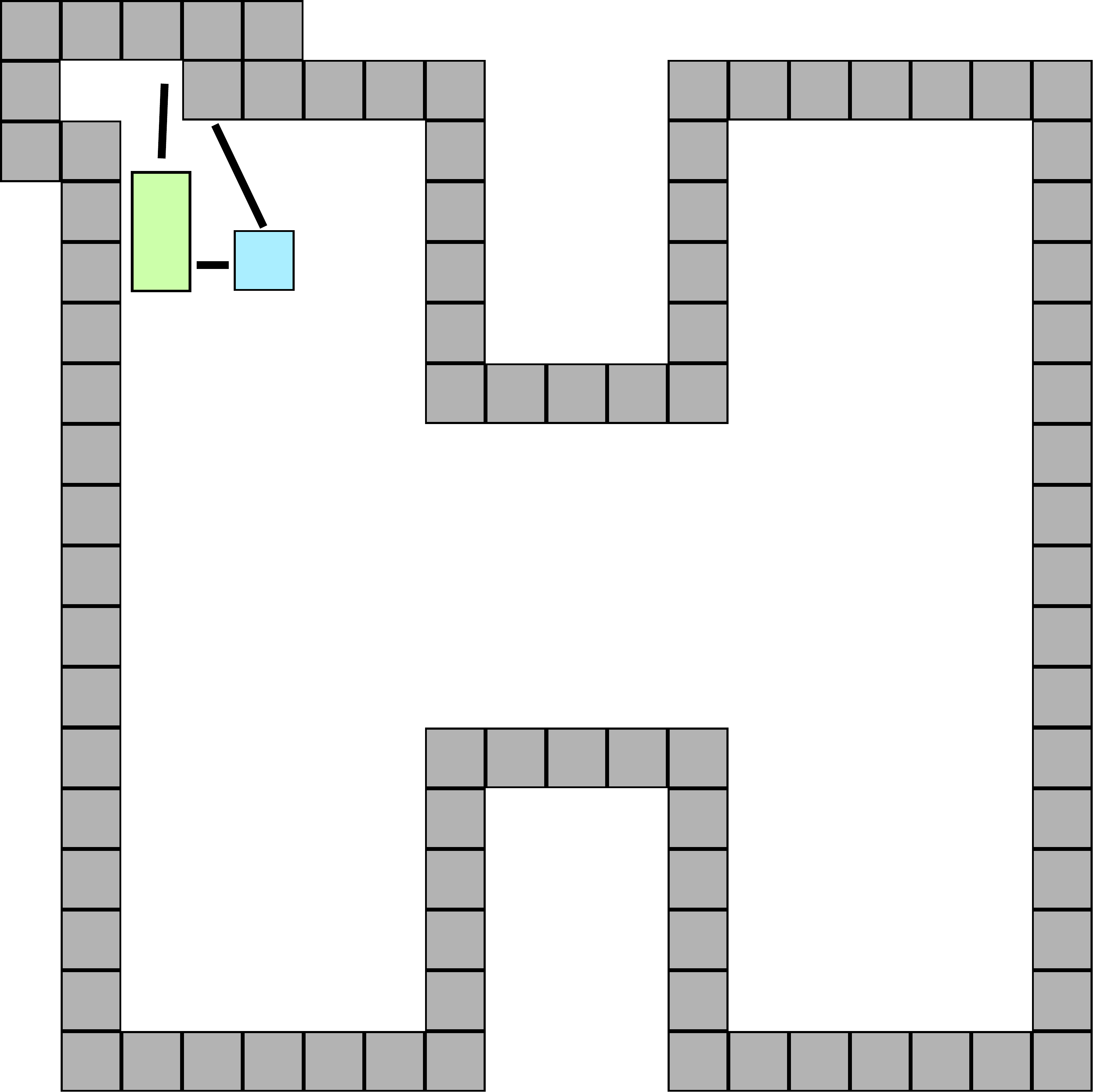}}
  \qquad %
  \subfloat[FRAME{[MOLD]} \label{fig-inner:b}]{\includegraphics[width=0.25\textwidth]{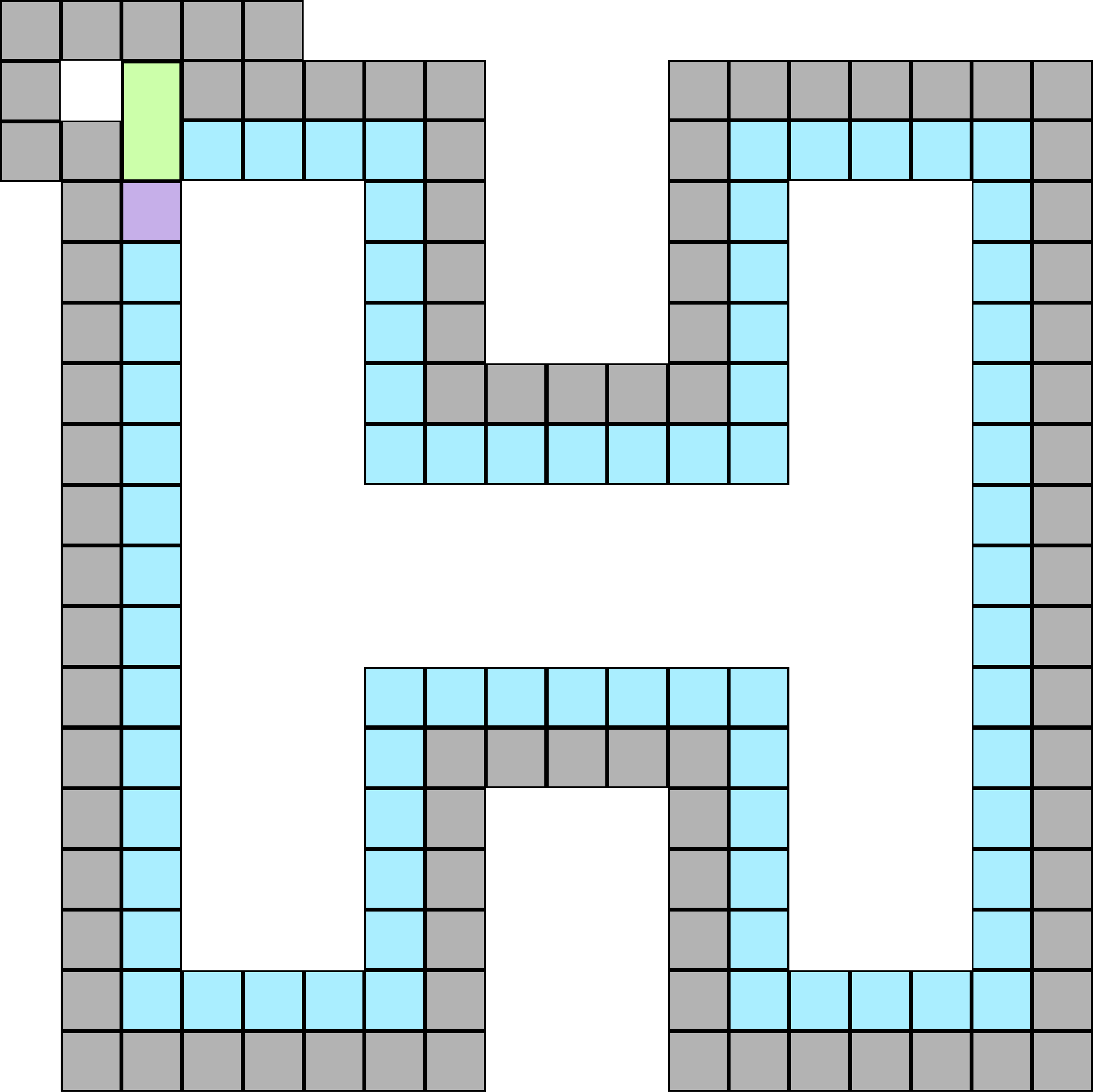}}
  \qquad %
  \subfloat[Completed drilling \label{fig-inner:c} ]{\includegraphics[width=0.25\textwidth]{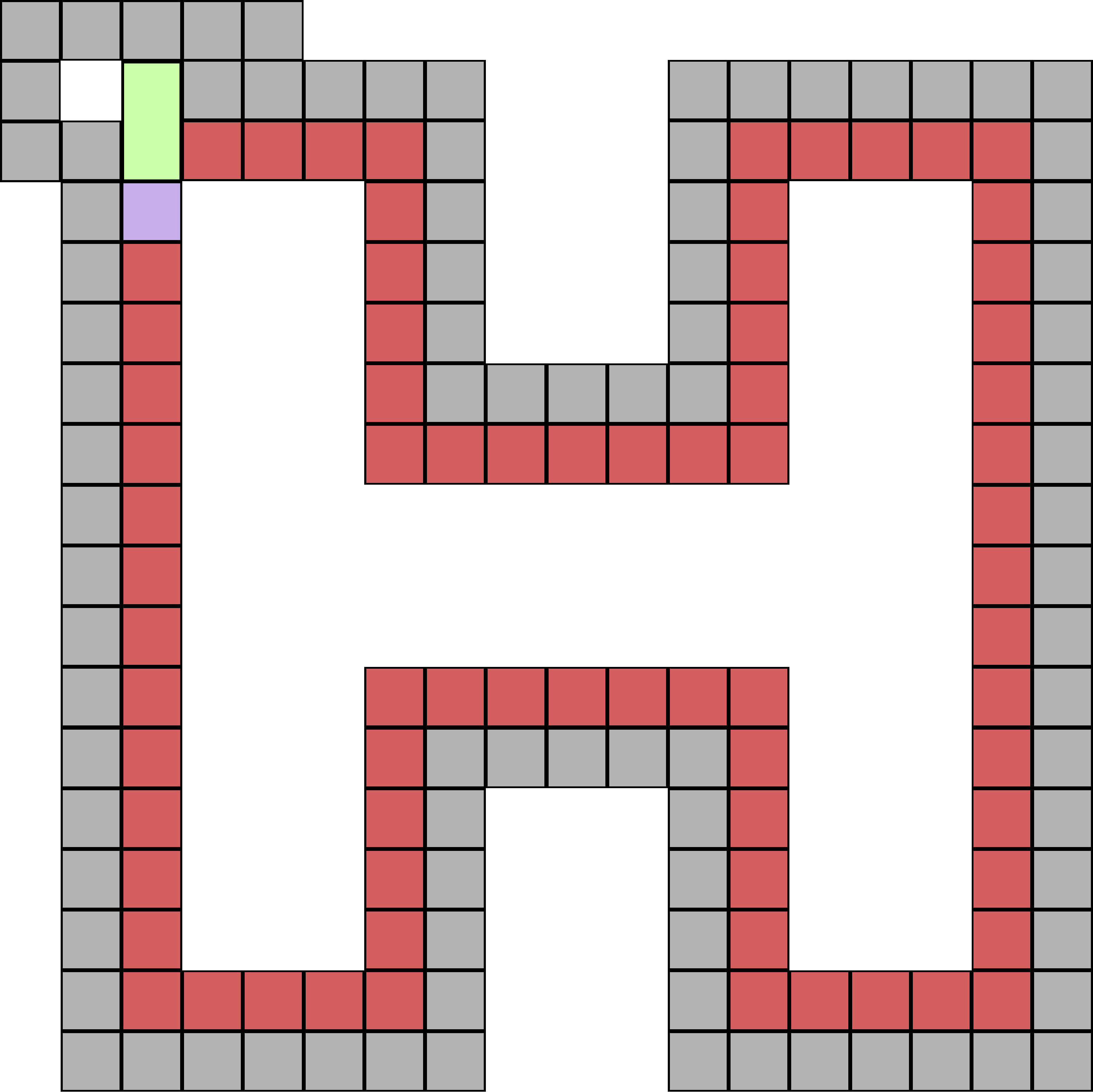}}
  \qquad %
  \subfloat[FRAME{[HOLLOW[$\Upsilon$]]} \label{fig-inner:d} ]{\includegraphics[width=0.25\textwidth]{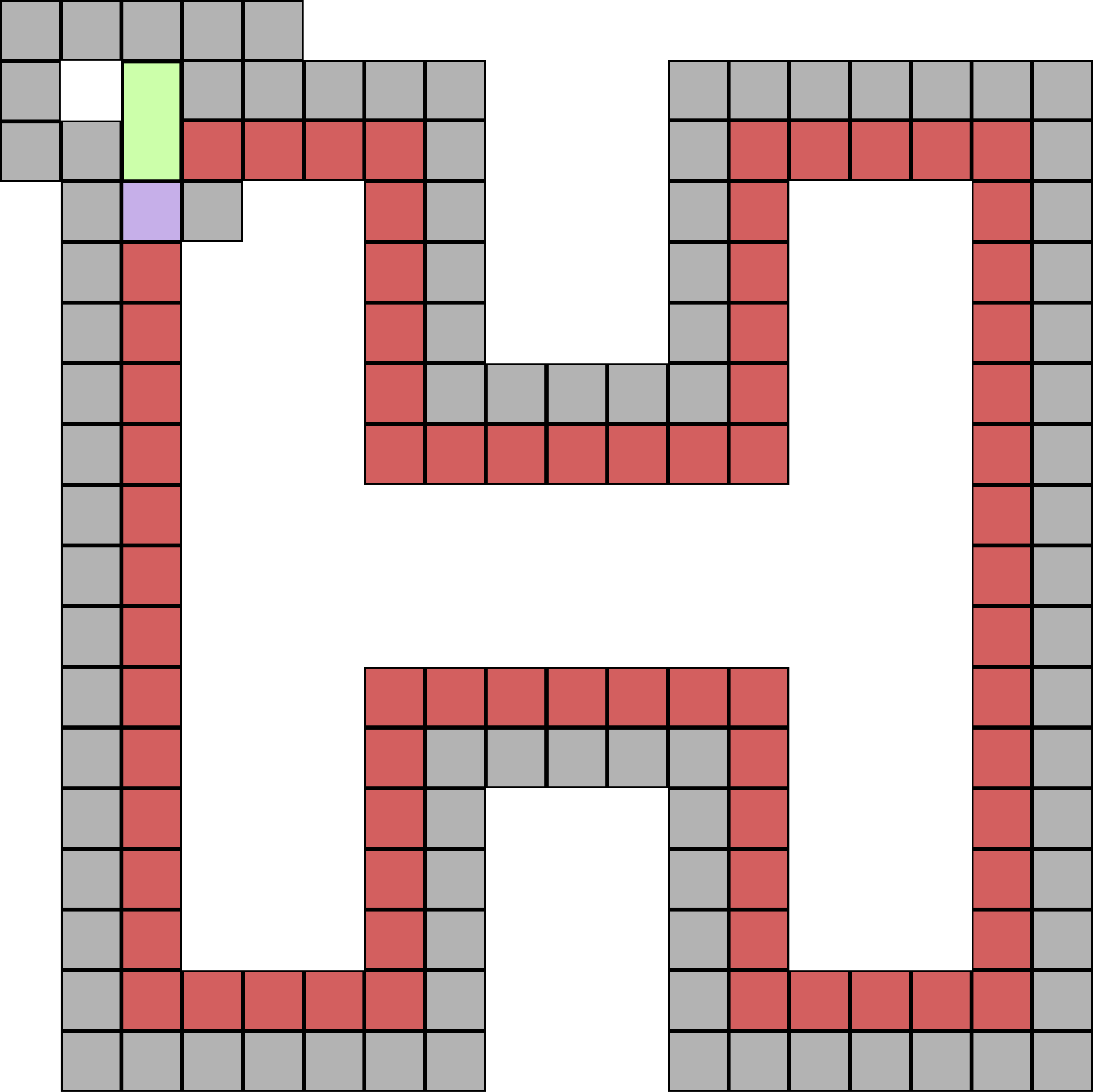}}
  \qquad %
  \subfloat[FRAME{[START]}\label{fig-inner:e} ]{\includegraphics[width=0.25\textwidth]{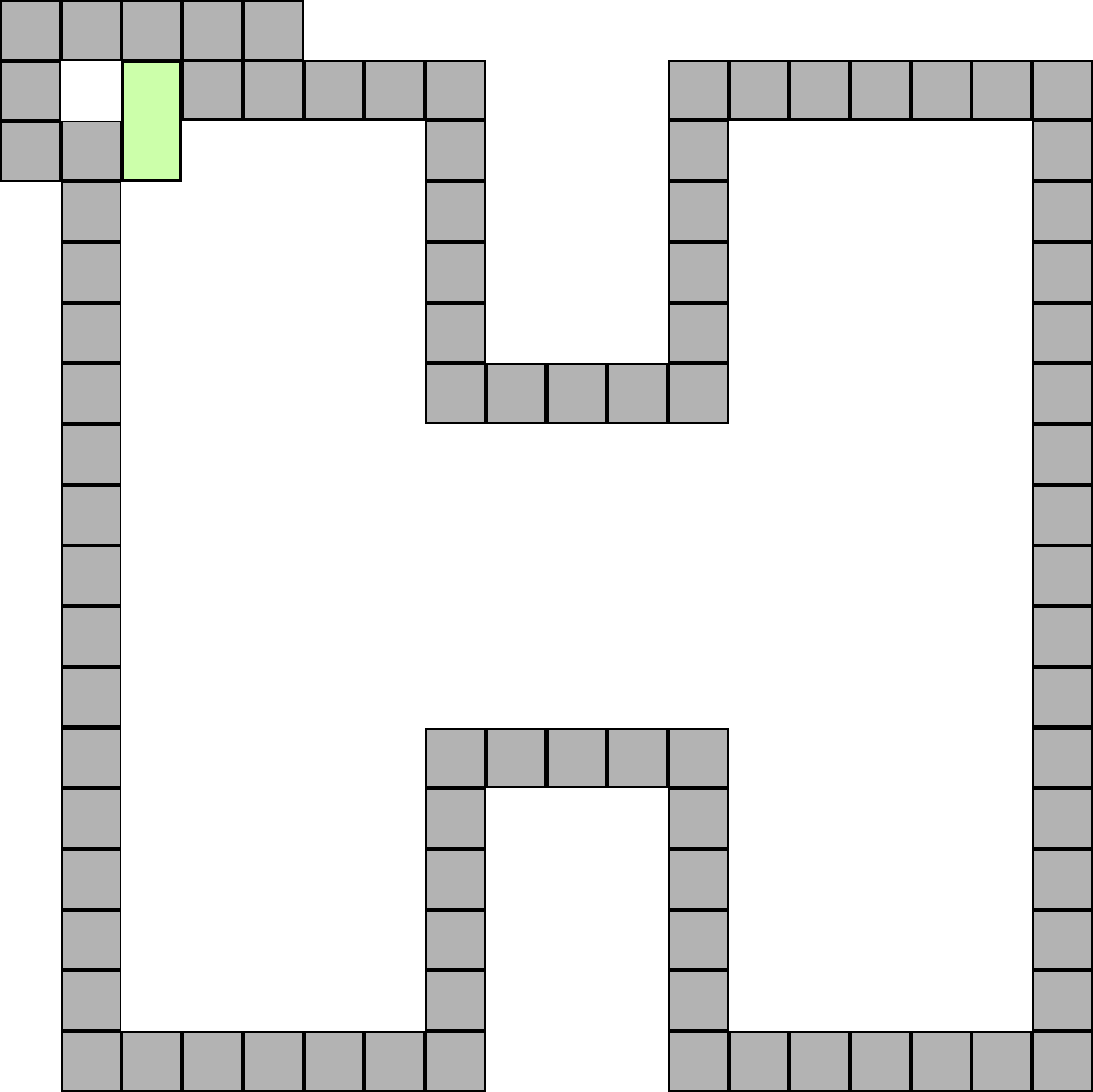}}
  \qquad %
  \subfloat[HOLLOW{[$\Upsilon$]} \label{fig-inner:f} ]{\includegraphics[width=0.25\textwidth]{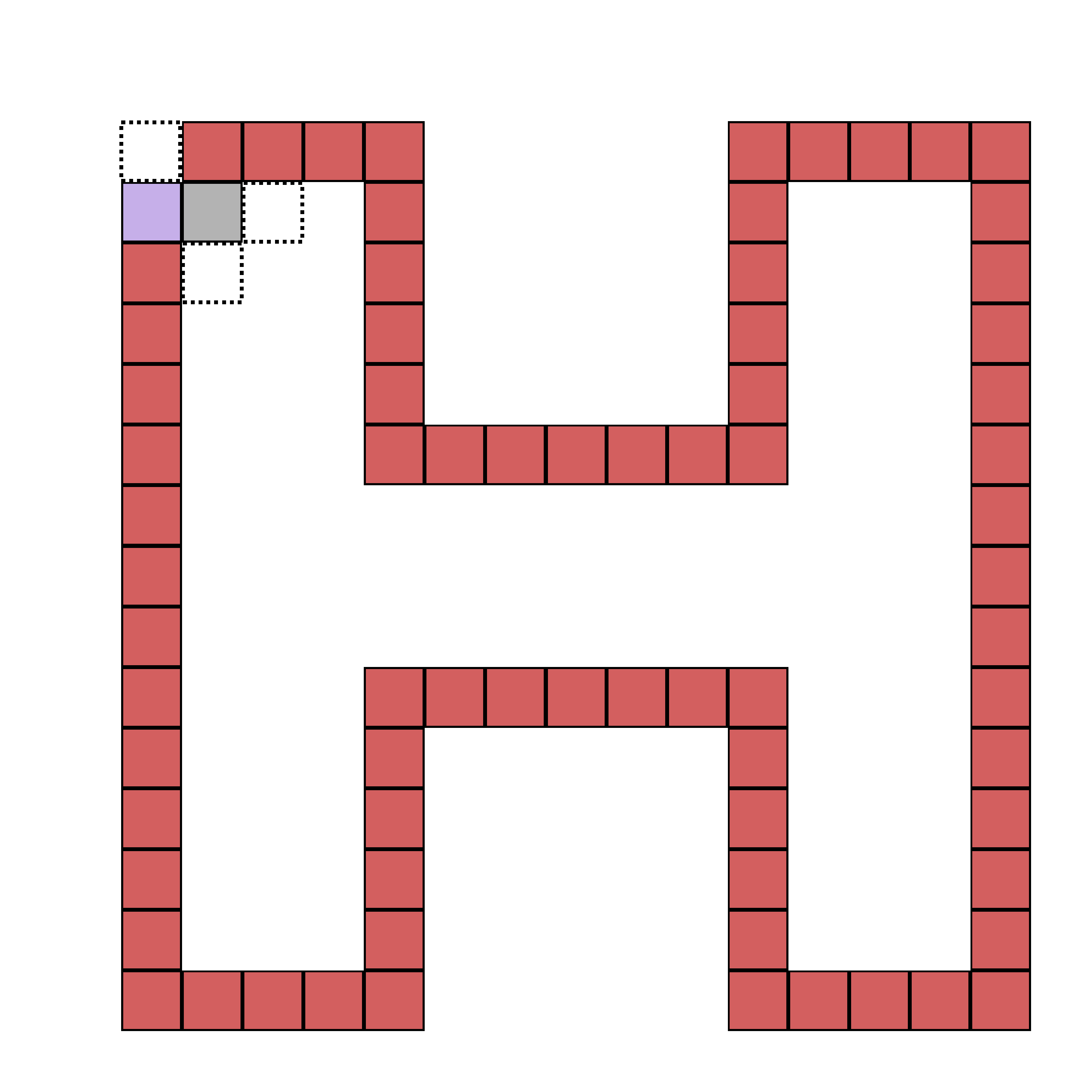}}
  \caption{High level process for making a copy of an input shape $\Upsilon$ from a \emph{frame} of $\Upsilon$.}
  \label{fig-inner}
\end{figure}

The high-level process for replicating a given input assembly $\Upsilon$ is described here. We assume that $\Upsilon$ has special glues along its perimeter such that all North, East, South, and West facing edges have glues of $N, E, S,$ and $W$, respectively, with the exception of the Northernmost-Westernmost unit of the shape having a North glue of $C1$ and West glue of $C2$ (Fig. \ref{fig-outer:a}). Intuitively, the process is to assemble an outline of $\Upsilon$ that is filled to have the same shape as $\Upsilon$. In Phase 1, \emph{mold gadgets} attach clockwise along the outside perimeter of $\Upsilon$, to detect the edges of the shape. In Phase 2, the mold gadgets are replaced, counterclockwise gadget-by-gadget, by \emph{drill gadgets}. These drill gadgets follow the path laid out by the mold gadgets and do not attach to the shape. The drill gadgets use negative glues to destabilize adjacent mold gadgets. This technique is used to create an assembly called the FRAME which outlines $\Upsilon$ but whose bounding box is too large to simply fill to get a copy of the shape of $\Upsilon$. In Phase 3, \emph{inner mold gadgets} detect the outline of the inside perimeter of the FRAME. In Phase 4, the inner mold gadgets are replaced, counterclockwise gadget-by-gadget, by \emph{inner drill gadgets} in order to begin destabilizing the FRAME from the assembly. With the help of an additional \emph{inner post-drill gadget}, the assembly destabilizes into the FRAME and HOLLOW[$\Upsilon$]. Using \emph{fill gadgets}, HOLLOW[$\Upsilon$] is filled to create COPY[$\Upsilon$], which is a copy of the original shape.

\vspace{5mm}
\textbf{Phase 1: Shape detection via mold gadgets (Section~\ref{sec:phase1Gadgets})}
\begin{enumerate}
  \setlength{\itemsep}{1pt}
  \setlength{\parskip}{0pt}
  \setlength{\parsep}{0pt}
  \item Place an \emph{outer start gadget} at the Northernmost-Westernmost corner utilizing glues $C1$ and $N$ (Fig.~\ref{fig-outer:b}).
  \item From the East edge of the outer start gadget, moving clockwise, trace the outside perimeter of the shape with a layer of \emph{mold gadgets}. Call the result MOLD[$\Upsilon$] (Figs. \ref{fig-outer:b}, \ref{fig-outer:c}).
\end{enumerate}

\textbf{Phase 2: Drilling to create a FRAME (Section~\ref{sec:phase2Gadgets})}
\begin{enumerate}
  \setlength{\itemsep}{1pt}
  \setlength{\parskip}{0pt}
  \setlength{\parsep}{0pt}
  \item Place a \emph{pre-drill} gadget at the Northernmost-Westernmost corner of MOLD[$\Upsilon$].
  \item Starting from the pre-drill gadget, and moving counter-clockwise, replace the mold gadgets in the surrounding layer with \emph{drill gadgets} that do not have affinity to the shape (Fig. \ref{fig-outer:d}).
  \item When done drilling, place a \emph{post-drill gadget}, call the result FRAME[$\Upsilon$] (Fig. \ref{fig-outer:e}).
  \item FRAME[$\Upsilon$] is unstable and separates into two assemblies, FRAME (Fig. \ref{fig-outer:f}) and START[$\Upsilon$]. START[$\Upsilon$] is an assembly that contains only the input shape with the start gadget attached.
\end{enumerate}

\textbf{Phase 3: Detecting the inside of the FRAME (Section~\ref{sec:phase3Gadgets})}
\begin{enumerate}
  \setlength{\itemsep}{1pt}
  \setlength{\parskip}{0pt}
  \setlength{\parsep}{0pt}
  \item Place an \emph{inner start gadget} at the Northernmost-Westernmost corner on the inside perimeter of FRAME, call the result FRAME[START] (Fig. \ref{fig-inner:a}).
  \item Starting from the inner start gadget, and moving clockwise, trace the inside perimeter of FRAME[START] with a layer of mold gadgets. Call the result FRAME[MOLD] (Fig. \ref{fig-inner:b}).
\end{enumerate}
\newpage
\textbf{Phase 4: Drilling and post-drilling to create HOLLOW[$\Upsilon$] (Section~\ref{sec:phase4Gadgets})}
\begin{enumerate}
  \setlength{\itemsep}{1pt}
  \setlength{\parskip}{0pt}
  \setlength{\parsep}{0pt}
  \item Starting from the Northernmost-Westernmost inside corner, replace the inner mold gadgets on the inside perimeter with inner drill gadgets that do not have affinity to FRAME (Fig. \ref{fig-inner:c}).
  \item Place an \emph{inner post-drill gadget}. Call the result FRAME[HOLLOW[$\Upsilon$]] (Fig. \ref{fig-inner:d}).
  \item FRAME[HOLLOW[$\Upsilon$]] is unstable and separates into FRAME[START] and HOLLOW[$\Upsilon$] (Fig. \ref{fig-inner:e} and Fig. \ref{fig-inner:f}).
\end{enumerate}

\textbf{Phase 5: Filling and Repeating (Section~\ref{sec:phase5Gadgets})}
\begin{enumerate}
  \setlength{\itemsep}{1pt}
  \setlength{\parskip}{0pt}
  \setlength{\parsep}{0pt}
  \item HOLLOW[$\Upsilon$] is filled to become COPY[$\Upsilon$], which is an assembly with the same shape as $\Upsilon$.
  \item START[$\Upsilon$], repeats the process over again, starting at the second step of Phase 1.
  \item FRAME[START] also repeats creating copies of the shape, starting at the second step of Phase 3.
\end{enumerate}

\section{Replication Gadgets}\label{sec:replicationGadgets}
% -*- root: ../main.tex -*-
In this section, we describe the assemblies, or gadgets, which are constructed from the tiles in the initial state of the replication system. These gadgets are designed to work in a temperature $\tau = 10$ system. There is a constant number of these distinct gadget types which are used to replicate arbitrarily sized input assemblies. In our figures, a black line perpendicular and in the middle of the edge of two adjacent tiles indicates a unique infinite strength bond (i.e. the strength of the glue is $\gg \tau$ such that no detachment events can occur in which these tiles are separated). Although each glue strength can be found in the figures and their captions, there is a full table of glue strengths in Table~\ref{tbl:gluetable}. First, we describe the gadgets used in the process of creating the layer of mold gadgets around the outer perimeter of the input assembly.

\begin{table}
\begin{center}
\begin{tabular}{ | c | c | c | c | }
    \hline
    Label   &   Strength    &   Label   &   Strength \\ \hline
    C1      &   2           &   C2      &   9        \\ \hline
    N,E,S,W,n,e,s,w &   9    &   T*,t*,O*,o*,L   &   1        \\ \hline
    B*,b*    &   5          &  K,k    &   3        \\ \hline
    X*,Y*,Z*,x*,y*,z*      &   9  &    D       &   -7       \\ \hline
    F*,f*      &   -2      & V*,U*,H*,J*,S*,v*,u*,h*,j*,s* & 9     \\ \hline
    R,r       &   2   &    A* &  10       \\ \hline
    g*      &   9   &  M* & 5         \\ \hline
     Q &  -2 &  q &  -5 \\ \hline
\end{tabular}
\caption{The glue strengths of each glue label in the replication system.}
\label{tbl:gluetable}
\end{center}
\end{table} 

\subsection{Phase 1 Gadgets for the Outer Mold}\label{sec:phase1Gadgets}
% -*- root: ../../main.tex -*-
Below we describe the gadgets used to implement Phase 1 and create MOLD[$\Upsilon$].

\begin{figure}
    \centering
	\subfloat[\label{fig:outerMoldGadget1}]{\includegraphics[width=.25\textwidth]{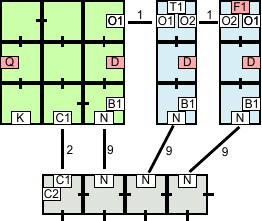}}
    \qquad %
    \subfloat[\label{fig:outerMoldGadget2}]{\includegraphics[width=.25\textwidth]{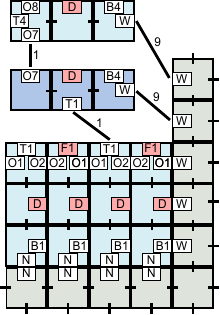}}
    \qquad %
    \subfloat[\label{fig:outerMoldGadget3}]{\includegraphics[width=.25\textwidth]{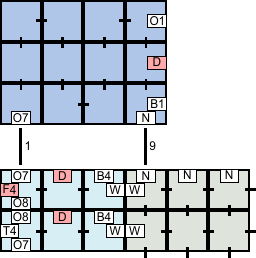}}
    \caption{(a) The start gadget attaches. $C1 + N = 2 + 9 \ge \tau=10.$, The first mold gadget attaches, cooperatively, to the start gadget and the input shape. $O1 + N = 1 + 9 \ge \tau.$ The second mold gadget attaches, cooperatively, to the previous mold gadget and the input shape. $O2 + N = 1 + 9 \ge \tau.$ (b) A concave corner of the assembly. The North$\rightarrow$West corner gadget attaches. $T1 + W = 1 + 9 \ge \tau.$ A mold gadget attaches. $O7 + W = 1 + 9 \ge \tau.$ (c) A convex corner of the assembly. The West$\rightarrow$North corner gadget attaches. $O7 + N = 1 + 9 \ge \tau.$ }
    \label{fig:outerMold}
\end{figure}

\textbf{Start gadget.} The start gadget is an assembly, designed to attach using the $C1$ glue and $N$ glue, which designates the Northernmost-Westernmost corner of the input assembly.
Once the start gadget attaches, a North mold gadget may attach (Fig. \ref{fig:outerMoldGadget1}).

\textbf{Mold gadgets.}
The mold gadgets, beginning at the start gadget as shown in Figure~\ref{fig:outerMoldGadget1}, walk along the input assembly in a clockwise manner to create MOLD[$\Upsilon$].
The North (South) mold gadgets are $1\times3$ assemblies designed to walk from West to East (East to West) along the North (South) edges of the input assembly.
The West (East) mold gadgets are $3\times1$ assemblies designed to walk from South to North (North to South) along the West (East) edges of the input assembly.
The mold gadgets expose either a positive or negative glue on their unused edge (e.g. North mold gadgets expose either a $T1$ or $F1$) used for detecting corners and drilling at corners, respectively.

\textbf{Corner mold gadgets.}
The \emph{corner mold gadgets} attach at corners of the input assembly once a mold gadget has reached the corner. Two sets of corner mold gadgets are used; one for concave corners and one for convex corners.
There are two types of concave and convex corner mold gadgets.
%shown in Figures \ref{fig:outerMoldGadget2} and \ref{fig:outerMoldType2Gadget2}
This is due to the edge length being even or odd. If the mold gadget that places adjacent to a corner of the shape has two negative glues, we denote the concave corner as a \emph{type-1} concave corner and otherwise as a \emph{type-2} concave corner.
%We define two sets of concave corner mold gadgets, with small differences, to handle both of these cases.
A type-1 concave corner gadget attachment can be seen in Figure~\ref{fig:outerMoldGadget2}, and convex in Figure~\ref{fig:outerMoldGadget3}.
Type-2 corner gadget details can be seen in Figure~\ref{fig:outerMoldType2}.
The North$\rightarrow$West (South$\rightarrow$East) corner gadget attaches cooperatively to a North (South) mold gadget and the input shape.
The East$\rightarrow$North (West$\rightarrow$South) corner gadget attaches cooperatively to a East (West) mold gadget and the input shape.

\begin{figure}[h]
    \centering
    \subfloat[\label{fig:outerMoldType2Gadget1}]{\includegraphics[width=.28\textwidth]{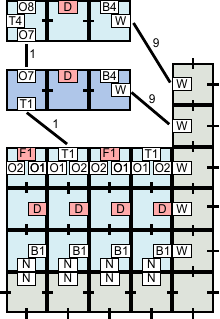}}
    \qquad %
    \subfloat[\label{fig:outerMoldType2Gadget2}]{\includegraphics[width=.28\textwidth]{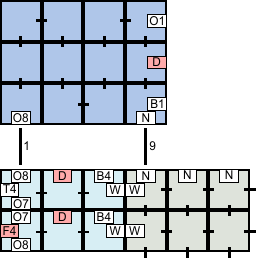}}
    \caption{(a) A type-2 concave corner of the assembly. The type-2 North$\rightarrow$West corner gadget attaches. $T1 + W = 1 + 9 \ge \tau.$ An mold gadget attaches. $O7 + W = 1 + 9 \ge \tau.$ (b) An convex corner of the assembly. The West$\rightarrow$North corner gadget attaches. $O8 + N = 1 + 9 \ge \tau.$}
    \label{fig:outerMoldType2}
\end{figure}

%\textbf{Convex corner mold gadgets.}
%Due to the edge length being even or odd, we also have 2 sets of convex corner mold gadgets. If the mold gadget that places adjacent to a convex corner of the shape has two negative glues, we denote the corner as a \emph{type-1} convex corner (Fig. \ref{fig:outerMoldGadget3}) and a \emph{type-2} convex corner (Fig. \ref{fig:outerMoldType2Gadget2}, otherwise. We define two sets of convex corner mold gadgets to handle both of these cases. The convex corner mold gadgets, of both types, (West$\rightarrow$North, North$\rightarrow$East, East$\rightarrow$South, and South$\rightarrow$West corner gadgets) are required for the mold process continue clockwise around convex corners of the input assembly. The West$\rightarrow$North corner mold gadget attaches cooperatively to a West mold gadget and the input assembly, then North mold gadget may attach cooperatively to the West$\rightarrow$North corner mold gadget and the input assembly.

\subsection{Phase 2 Gadgets for Making the FRAME}\label{sec:phase2Gadgets}
% -*- root: ../../main.tex -*-
\begin{figure}
    \centering
    \subfloat[\label{fig:predrillingGadget1}]{\includegraphics[width=.25\textwidth]{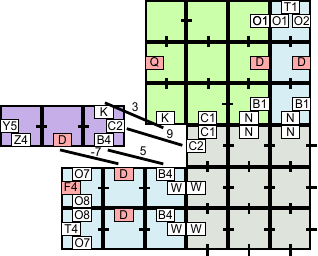}}
    \qquad %
    \subfloat[\label{fig:predrillingGadget2}]{\includegraphics[width=.25\textwidth]{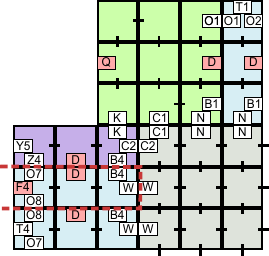}}
    \qquad %
    \subfloat[\label{fig:predrillingGadget3}]{\includegraphics[width=.25\textwidth]{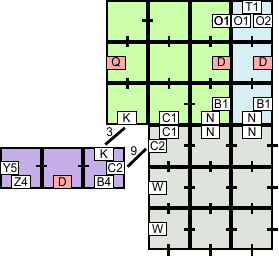}}
    \caption{The pre-drill process. (a) Case 1: The mold process has completed. The pre-drill gadget attaches cooperatively to the start gadget, input shape, and a mold gadget. $K + C2 + B4 + D = 3 + 9 + 5 - 7 \ge \tau.$ (b) The pre-drill gadget destabilizes the adjacent mold gadget. $B4 + W + O8 + D = 5 + 9 + 1 - 7 < \tau.$ (c) Case 2: The mold process has not completed, the pre-drill gadget may still attach. After attachment, the negative `D' glue prevents a mold gadget from occupying the space directly below. $K + C2 = 3 + 9 \ge \tau.$}
    \label{fig:predrillingGadgets}
\end{figure}

Below we describe the gadgets used to implement Phase 2 and create FRAME and START[$\Upsilon$].

\textbf{Pre-drill gadget.} A gadget that binds to the start gadget, a mold gadget (in some cases) , and the input shape that is designed to allow the placement of a \emph{West drill helper} to its South. It allows the drilling process to begin once mold gadgets have finished encircling the input shape (Fig. \ref{fig:predrillingGadgets}).

% -*- root: ../../main.tex -*-
% The outer drill gadgets may attach once the outer mold gadgets have circled the input assembly. Each time a drill gadget is attached, a negative glue interaction causes the outer mold gadgets to become breakable from the input assembly.

% The outer mold blocks, beginning at the start block as shown in Figure~\ref{fig:outerMoldGadget1}, walk along the input assembly in a clockwise manner to create the outer mold.
% The North (South) outer mold blocks are $1\times5$ assemblies designed to walk from West to East (East to West) along the North (South) edges of the input assembly.
% The West (East) outer mold blocks are $5\times1$ assemblies designed to walk from South to North (North to South) along the West (East) edges of the input assembly.
% The outer mold blocks are able attach along acute corners of the input assembly, as shown in Figure~\ref{fig:outerMoldGadget2}.
\begin{figure}
    \centering
    \subfloat[\label{fig:drillingGadgets1}]{\includegraphics[width=.25\textwidth]{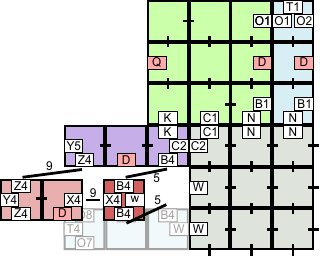}}
    \qquad %
    \subfloat[\label{fig:drillingGadgets2}]{\includegraphics[width=.25\textwidth]{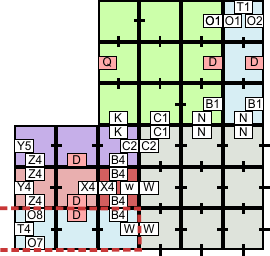}}
    \qquad %
    \subfloat[\label{fig:drillingGadgets3}]{\includegraphics[width=.25\textwidth]{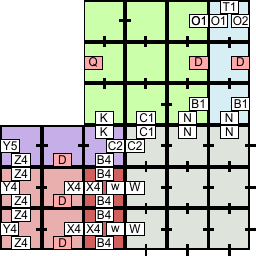}}
    \caption{The drilling process begins. (a) A West drill helper attaches cooperatively to the pre-drill gadget and a mold gadget. $B4 + B4 = 5 + 5 \ge \tau.$ A drill gadget attaches cooperatively to the pre-drill gadget and drill helper. $Z4 + X4 + D = 9 + 9 - 7 \ge \tau.$ (b) The drill gadget destabilizes the mold gadget underneath. $B4 + W + O7 + D = 5 + 9 + 1 - 7 < \tau.$ (c) The removal of the mold gadget allows the process shown in (a) and (b) to repeat. The process continues along the face of the shape.}
    \label{fig:drilling}
\end{figure}

\textbf{Drill gadgets.} The drill gadgets (Fig. \ref{fig:drillingGadgets1}), along with drill helpers, walk along the boundary of the input assembly in a counter-clockwise manner destabilizing mold gadgets and replacing them to create the FRAME. The drill gadgets do not attach to the input shape $\Upsilon$, so when the drills have removed all of the mold gadgets, the FRAME will not be attached to the input shape. The drill gadgets are assemblies that attach to the North (South) drill helpers and destabilize mold gadgets to their East (West). The West (East) drill gadgets are assemblies that attach to the West (East) drill helpers and destabilize mold gadgets to their North (South).

\textbf{Drill helpers.} The drill helpers (Fig. \ref{fig:drillingGadgets1}) are $1\times1$ assemblies that, beginning at the pre-drill gadget, walk along the boundary of the input assembly in a counter-clockwise manner exposing glues that allow drill gadgets to destabilize adjacent mold gadgets. These gadgets walk along the boundary by following glues exposed by mold gadgets, rather than using glues on the input assembly.

\textbf{Drill corner gadgets.} Drill corner gadgets bind cooperatively at type-1  and type-2 corners  to drill helpers and mold gadgets in order for the drilling process to turn corners.
%Drill corner gadget details are omitted for space, but work similarly to non-corner drill gadgets.
Details can be seen in Figures~\ref{fig:drillingAcuteType1},\ref{fig:drillingAcuteType2},and \ref{fig:drillingObtuse}.

\begin{figure}[h]
    \centering
    %\subfloat[\label{fig:drillingAcuteType1-1}]{\includegraphics[width=.28\textwidth]{gadgets/outer_fig/drillingacutecornerType2}}
    %\qquad %
    %\subfloat[\label{fig:drillingAcuteType1-2}]{\includegraphics[width=.28\textwidth]{gadgets/outer_fig/drillingacutecorner2Type2}}
    %\qquad %
    \subfloat[\label{fig:drillingAcuteType1-3}]{\includegraphics[width=.28\textwidth]{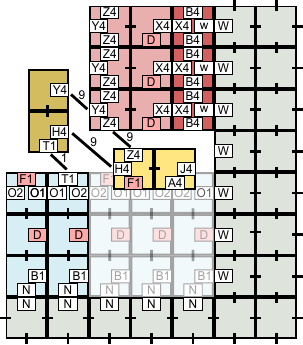}}
    \qquad %
    \subfloat[\label{fig:drillingAcuteType1-4}]{\includegraphics[width=.28\textwidth]{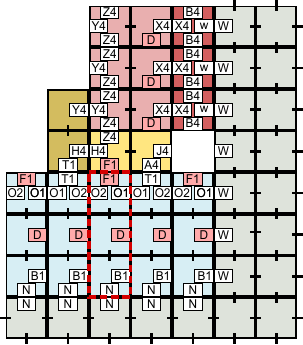}}
    \qquad %
    \subfloat[\label{fig:drillingAcuteType1-5}]{\includegraphics[width=.28\textwidth]{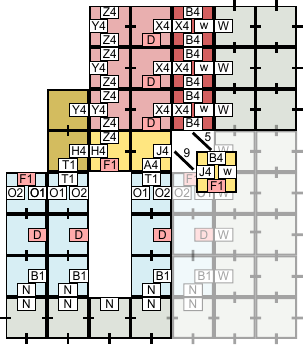}}
    \qquad %
    \subfloat[\label{fig:drillingAcuteType1-6}]{\includegraphics[width=.28\textwidth]{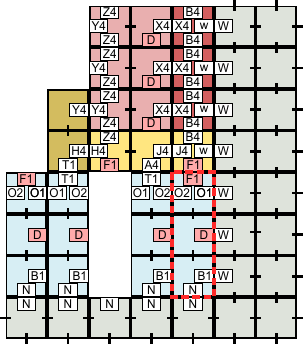}}
    \qquad %
    \subfloat[\label{fig:drillingAcuteType1-7}]{\includegraphics[width=.28\textwidth]{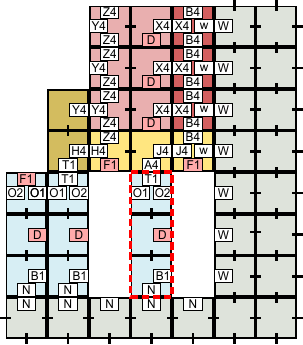}}
    \qquad %
    \subfloat[\label{fig:drillingAcuteType1-8}]{\includegraphics[width=.28\textwidth]{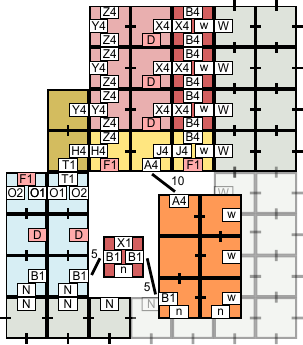}}
    %\qquad %
    %\subfloat[\label{fig:drillingAcuteType1-9}]{\includegraphics[width=.28\textwidth]{gadgets/outer_fig/drillingacutecorner9Type2}}
    \caption{The drilling process encounters a type-1 concave corner of the assembly. (a) A drill gadget destabilizes the North$\rightarrow$West corner gadget. $B4 + W + T1 + D = 5 + 9 + 1 - 7 < \tau.$  A West$\rightarrow$North drill corner helper attaches to a mold gadget and a drill gadget. $Y4 + T1 = 9 + 1 \ge \tau.$ A West$\rightarrow$North drill corner gadget attaches, cooperatively. $H4 + Z4 + F1 = 9 + 9 - 2 \ge \tau.$ (b) The West$\rightarrow$North drill corner gadget destabilizes a mold gadget to its South. $O1 + N + O2 + F1 = 1 + 9 + 1 - 2 < \tau.$ (c) A mold gadget has destabilized and has broken off the assembly. Another West$\rightarrow$drill gadget attaches. (d) destabilizes a mold gadget to its South. $N + O2 + F1 = 9 + 1 - 2 < \tau$ (e) Without any adjacent mold gadgets to its East or West, a mold gadget becomes unstable. $N = 9 < \tau.$ (f) A drill corner gadget fills in the space in the concave corner and exposes a glue that will allow drilling to continue to the West. $A4 = 10 \ge \tau.$}
    \label{fig:drillingAcuteType1}
\end{figure}

% outer drill acute corner 2
\begin{figure}[h]
    \centering
    \subfloat[\label{fig:drillingAcuteType2-1}]{\includegraphics[width=.28\textwidth]{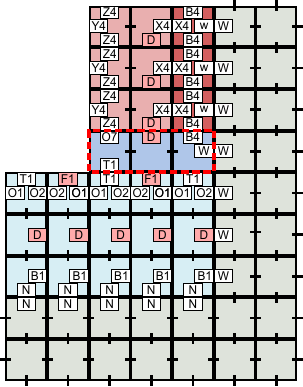}}
    \qquad %
    \subfloat[\label{fig:drillingAcuteType2-2}]{\includegraphics[width=.28\textwidth]{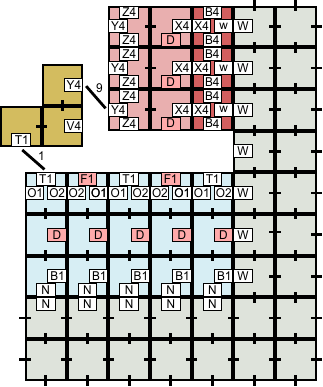}}
    \qquad %
    \subfloat[\label{fig:drillingAcuteType2-3}]{\includegraphics[width=.28\textwidth]{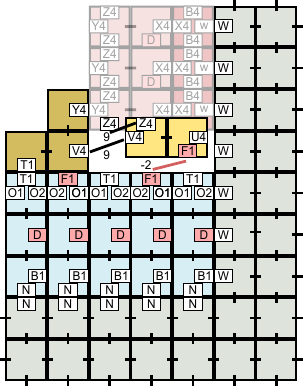}}
    \qquad %
    \subfloat[\label{fig:drillingAcuteType2-4}]{\includegraphics[width=.28\textwidth]{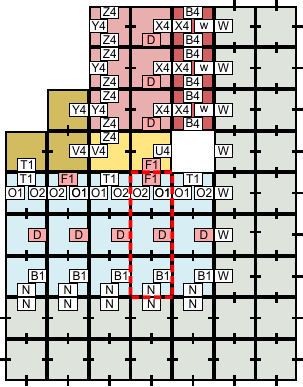}}
    \qquad %
    \subfloat[\label{fig:drillingAcuteType2-5}]{\includegraphics[width=.28\textwidth]{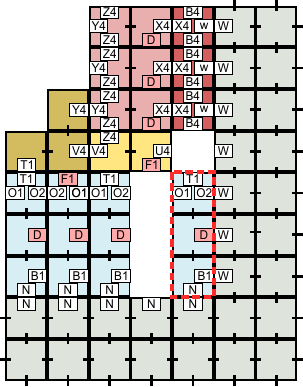}}
    \qquad %
    \subfloat[\label{fig:drillingAcuteType2-6}]{\includegraphics[width=.28\textwidth]{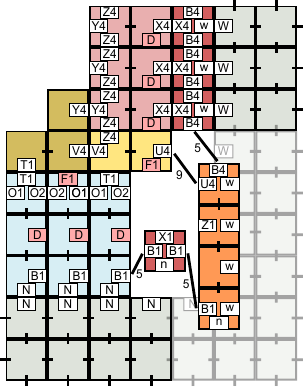}}
    \caption{The drilling process encounters a type-2 concave corner of the assembly. (a) A drill gadget destabilizes the North$\rightarrow$West corner gadget. $B4 + W + T1 + D = 5 + 9 + 1 - 7 < \tau$ (b) A West$\rightarrow$North drill corner helper attaches to a mold gadget and a drill gadget. $Y4 + T1 = 9 + 1 \ge \tau.$ (c) A West$\rightarrow$North drill corner gadget attaches, cooperatively. $V4 + Z4 + F1 = 9 + 9 - 2 \ge \tau.$ (d) The West$\rightarrow$North drill corner gadget destabilizes a mold gadget to its South. $O1 + N + O2 + F1 = 1 + 9 + 1 - 2 < \tau.$ (e) Without a mold gadget to its West, the mold gadget in the concave corner is unstable. $N = 9 < \tau.$ (f) A drill corner gadget fills in the space in the concave corner. $U4 + B4 = 9 + 5 \ge \tau.$ This gadget exposes a glue on its West side that allows drilling to continue to the West. $B1 + B1 = 5 + 5 \ge \tau.$}
    \label{fig:drillingAcuteType2}
\end{figure}

% outer drill obtuse
\begin{figure}[h]
    \centering
    \subfloat[\label{fig:drillingObtuse1}]{\includegraphics[width=.28\textwidth]{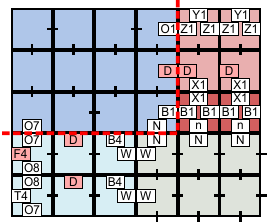}}
    \qquad %
    \subfloat[\label{fig:drillingObtuse2}]{\includegraphics[width=.28\textwidth]{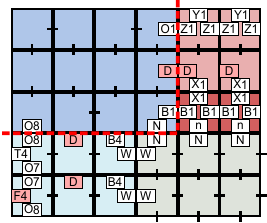}}
    \qquad %
    \subfloat[\label{fig:drillingObtuse3}]{\includegraphics[width=.28\textwidth]{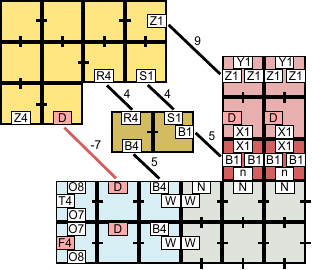}}
    \caption{The drilling process encounters convex corners of the assembly. (a) Coming from the East, drill gadgets encounter a type-1 convex corner of the assembly. The mold corner gadget becomes unstable and will break from the assembly. $O7 + N + B1 + D = 1 + 9 + 5 - 7 \ge \tau.$ (b) The drill gadgets encounter a type-2 convex corner gadget. It becomes unstable and will break from the assembly. $O8 + N + B1 + D = 1 + 9 + 5 - 7 < \tau.$ (c) Convex corner gadgets bind and allow the drilling process to continue. $B4 + B1 = 5 + 5 \ge \tau.$ $Z1 + R4 + S1 + D = 9 + 4 + 4 - 7 \ge \tau.$}
    \label{fig:drillingObtuse}
\end{figure} 

% \paragraph{Drill blocks.} The first type of outer drill block is a $1\times3$ block that attaches cooperatively via the $X_1$ and $O$ glues. The $X_1$ glue is exposed when a \emph{upper outer drill helper} tile is placed. The second type of outer drill block is also a $1\times3$ block that attaches coopertively via the $X_2$ and $L$ glues. The $L$ glue is exposed when and \emph{lower outer drill helper} tile is placed.

% -*- root: ../../main.tex -*-
\begin{figure}
    \centering
    \subfloat[\label{fig:postdrillingGadget1}]{\includegraphics[width=.25\textwidth]{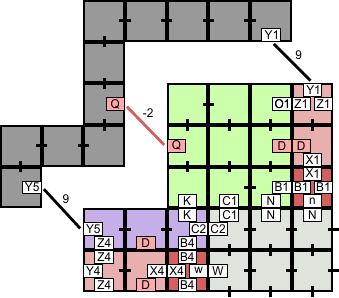}}
    \qquad %
    \subfloat[\label{fig:postdrillingGadget2}]{\includegraphics[width=.25\textwidth]{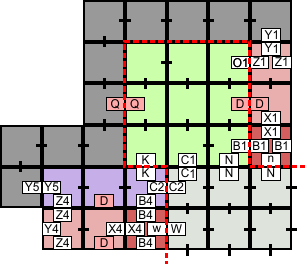}}
    \qquad %
    \subfloat[\label{fig:postdrillingGadget3}]{\includegraphics[width=.25\textwidth]{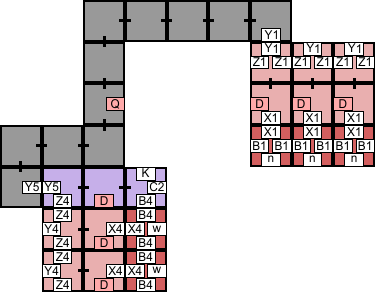}}
    \caption{The post-drill process begins. (a) The post-drill gadget binds at the Northernmost-Westernmost corner of the assembly that contains the input shape. It binds cooperatively to the last placed drill block and the pre-drill gadget. $Y5 + Y1 + Q = 9 + 9 - 2 \ge \tau.$ (b) Once placed, the pre-drill gadget destabilizes a cut which includes the input shape and the start block. $B1 + C2 + K + Q + D = 5 + 9 + 3 - 2 - 7 < \tau.$ (c) Once the assembly with the input shape and the start block detached, the FRAME remains.}
    \label{fig:postdrillingGadgets}
\end{figure}

\textbf{Post-drill gadget.} \emph{Post-drilling} refers to the the process of removing the input shape to create a FRAME. It is needed because the pre-drill gadget is still attached to the input shape. The post-drill gadget may attach once the drill gadgets have encircled the input shape. It attaches to the pre-drill gadget and last-placed drill gadget. Once attached, it destabilizes the assembly with the input shape and start gadget to create the FRAME. (See Fig. \ref{fig:postdrillingGadgets}.)

\subsection{Phase 3 Gadgets for the Inner Mold}\label{sec:phase3Gadgets}
% -*- root: ../../main.tex -*-
\begin{figure}
    \centering
	\subfloat[\label{fig:innerMoldGadget1}]{\includegraphics[width=.25\textwidth]{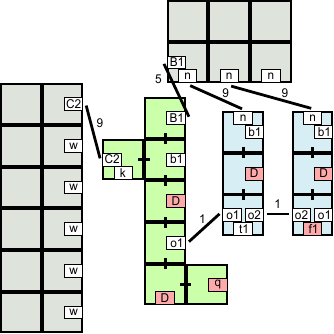}}
    \qquad %
    \subfloat[\label{fig:innerMoldGadget2}]{\includegraphics[width=.25\textwidth]{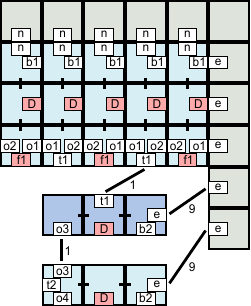}}
    \qquad %
    \subfloat[\label{fig:innerMoldGadget3}]{\includegraphics[width=.25\textwidth]{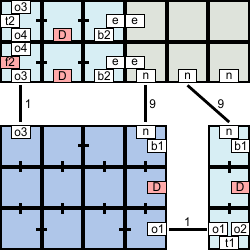}}
    \caption{Starting the inner mold process and handling type-1 corners. (a) The inner start gadget attaches to the inside perimeter of the FRAME. $C2 + B1 = 9 + 5 \ge \tau.$ An inner mold gadget attaches cooperatively to the start gadget and the FRAME. $o1 + n = 1 + 9 \ge \tau.$ Another inner mold gadget attaches cooperatively to the FRAME and the previous inner mold gadget. $o2 + n = 1 + 9 \ge \tau.$ (b) A concave corner of the assembly. A South$\rightarrow$West concave corner gadget attaches to the concave corner.  $t1 + e = 1 + 9 \ge \tau.$ An inner mold gadget attaches cooperatively to the FRAME and the South$\rightarrow$West concave corner gadget. $o3 + e = 1 + 9 \ge \tau.$ (c) A convex corner of the assembly. A West$\rightarrow$South convex corner gadget attaches at the convex corner, allowing the mold process to continue around the convex corner. $o3 + n = 1 + 9 \ge \tau.$ $o1 + n = 1 + 9 \ge \tau.$}
    \label{fig:innerMold}
\end{figure}

\begin{figure}[h]
    \centering
    \subfloat[\label{fig:innerMoldGadgetType2-1}]{\includegraphics[width=.28\textwidth]{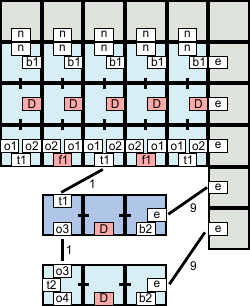}}
    \qquad %
    \subfloat[\label{fig:innerMoldGadgetType2-2}]{\includegraphics[width=.28\textwidth]{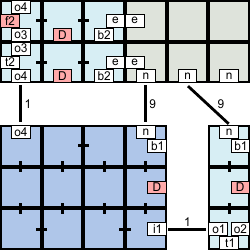}}
    \caption{Type-2 concave and convex corners of the assembly. (a) A type-2 South$\rightarrow$West concave corner gadget binds, which allows the mold process to continue around the concave corner. $t1 + e = 1 + 9 \ge \tau.$ $o3 + e = 1 + 9 \ge \tau.$(b) A type-2 West$\rightarrow$South corner gadget binds allowing the mold process to continue around the convex corner. $o4 + n = 1 + 9 \ge \tau.$ $o1 + n = 1 + 9 \ge \tau.$}
    \label{fig:innerMoldType2}
\end{figure} 

Here we describe the gadgets used to implement Phase 3 and create FRAME[MOLD]. The \emph{inner mold gadgets} are similar to the mold gadgets used in Section \ref{sec:phase1Gadgets}, with changes to allow them to encircle the inside perimeter of the FRAME, rather than the outside perimeter of the input shape $\Upsilon$. The Phase 3 gadgets (shown in Figure~\ref{fig:innerMold}) include an \emph{inner start gadget} which attaches to the FRAME to allow the attachment of inner mold gadgets which attach along the inside perimeter of the FRAME in a clockwise manner. The inner mold gadgets reflect the form and function of the mold gadgets shown in Section~\ref{sec:phase1Gadgets}, handling corners in a similar way. Type-2 inner mold corner gadgets can be seen in Figure~\ref{fig:innerMoldType2}.

\subsection{Phase 4 Gadgets for Making the Hollow Outline}\label{sec:phase4Gadgets}
% -*- root: ../../main.tex -*-
\begin{figure}
    \centering
    \subfloat[\label{fig:innerPreDrillingGadgets1}]{\includegraphics[width=.25\textwidth]{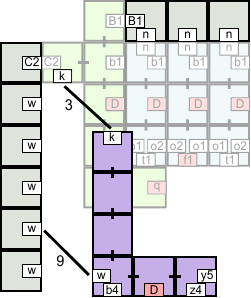}}
    \qquad %
    \subfloat[\label{fig:innerPreDrillingGadgets2}]{\includegraphics[width=.25\textwidth]{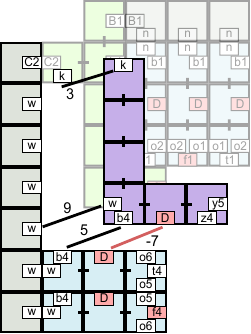}}
    \qquad %
    \subfloat[\label{fig:innerPreDrillingGadgets3}]{\includegraphics[width=.25\textwidth]{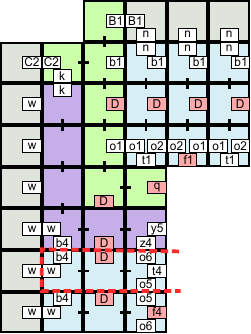}}
    \caption{\emph{The inner pre-drilling process begins. (a)} Case 1: The inner mold process has not completed and the inner pre-drill gadget binds cooperatively to the FRAME and the inner start gadget. $w + k = 9 + 3 \ge \tau.$ (b) Case 2: The inner mold process has completed and the inner pre-drill gadget binds cooperatively to the FRAME, inner start gadget, and an inner mold gadget. $w + k + b4 + D = 9 + 3 + 5 - 7 \ge \tau.$ (c) The inner pre-drill gadget prevents any inner mold gadget from attaching to its south and destabilizes any inner mold gadget that may have been in that space. $b4 + o5 + w + D = 5 + 1 + 9 - 7 < \tau.$}
    \label{fig:innerPreDrilling}
\end{figure}

\begin{figure}[h]
    \centering
    \subfloat[\label{fig:innerDrillingGadgets1}]{\includegraphics[width=.25\textwidth]{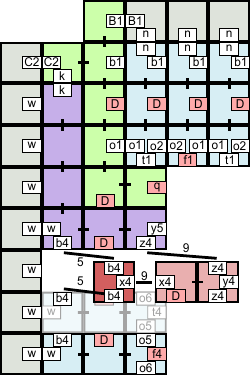}}
    \qquad %
    \subfloat[\label{fig:innerDrillingGadgets2}]{\includegraphics[width=.25\textwidth]{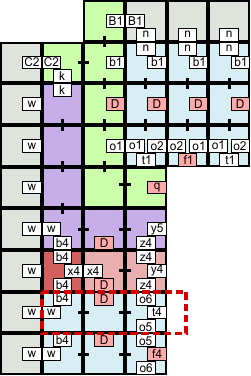}}
    \qquad %
    \subfloat[\label{fig:innerDrillingGadgets3}]{\includegraphics[width=.25\textwidth]{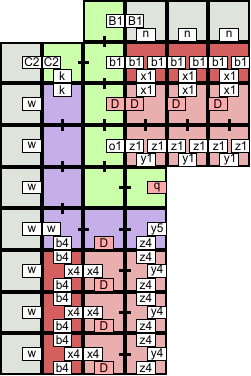}}
    \caption{The inner drilling process begins. (a) A drill helper binds to the inner pre-drill gadget and an inner mold gadget. $b4 + b4 = 5 + 5 \ge \tau.$ A drill block binds cooperatively to the drill helper and pre-drill gadget. $x4 + z4 + D = 9 + 9 - 7 \ge \tau.$ (b) The drill gadgets destabilize an inner mold gadget. $b4 + o5 + w + D = 5 + 1 + 9 - 7 < \tau.$ (c) The drilling process continues counter-clockwise along the inside perimeter of the FRAME until the drill gadgets reach the inner start gadget.}
    \label{fig:innerDrilling}
\end{figure}

\begin{figure}[t]
    \vspace*{-.3cm}
    \centering
    \subfloat[\label{fig:innerPostDrillingGadgets1}]{\includegraphics[width=.25\textwidth]{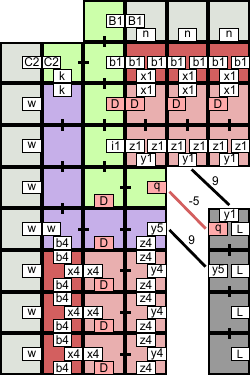}}
    \hspace*{3cm}
    \subfloat[\label{fig:innerPostDrillingGadgets2}]{\includegraphics[width=.25\textwidth]{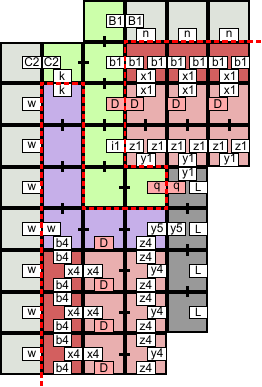}}
    \caption{The inner post-drilling process. (a) The inner post-drill gadget binds cooperatively to the inner pre-drill gadget and the second-to-last placed drill gadget. $y5 + y1 + q = 9 + 9 - 5 \ge \tau.$ (b) The assembly with the FRAME and the start gadget becomes unstable and breaks off the assembly. $b1 + w + k + q + D = 5 + 9 + 3 - 5 - 7 < \tau.$}
    \label{fig:innerPostDrilling}
    \vspace*{-.3cm}
\end{figure}

Below we describe the gadgets used to implement Phase 4 and create FRAME[START] and HOLLOW[$\Upsilon$]. The \emph{inner pre-drill gadget} can attach to the FRAME after the inner start gadget has placed, similar to the pre-drill gadget shown in Section~\ref{sec:phase2Gadgets}. The inner pre-drill gadget can be seen in Figure \ref{fig:innerPreDrilling}. Once the inner pre-drilling gadget has attached, the inner drill gadgets attach along the inside edge of the frame, removing the inner mold gadgets one by one counterclockwise. These gadgets reflect the form and function of the drill gadgets shown in Section~\ref{sec:phase2Gadgets}. 
The inner drilling process reflects that of the outer drilling, and can be seen in Figure~\ref{fig:innerDrillingGadgets1}. 
Corners are handled similarly to the \emph{outer} drill gadgets, from Phase 2 (Section~\ref{sec:phase2Gadgets}) An. When the drills completely attach to the inner perimeter of the FRAME, the \emph{inner post-drill} gadget can attach, as seen in Figure~\ref{fig:innerPostDrillingGadgets1}, to create FRAME[HOLLOW[$\Upsilon$]]. After the inner post-drill gadget attaches, the assembly destabilizes into FRAME[START] and HOLLOW[$\Upsilon$], as shown in Figure~\ref{fig:innerPostDrillingGadgets2}.

\subsection{Phase 5 Gadgets to Fill the Hollow Outline}\label{sec:phase5Gadgets}

Below we describe the gadgets used to implement Phase 5 to fill HOLLOW[$\Upsilon$] to create COPY[$\Upsilon$]. The \emph{filler gadgets} begin the process of filling the shape, shown in Figure~\ref{fig:innerFilling}. Two regions of HOLLOW[$\Upsilon$] need to be filled to become COPY[$\Upsilon$], the Northernmost-Westernmost corner (Figure~\ref{fig:innerPostDrillingGadgets3}) and the interior (Figures~\ref{fig:innerFilling1},\ref{fig:innerFilling2}).

% FRAME[START] allows for inner mold gadgets to again circle the inside perimeter of the shape, drill the mold gadgets, and make another hollow copy of the shape $\Upsilon$ an unbounded number of times.

\begin{figure}[t!]
    \centering
    \subfloat[\label{fig:innerPostDrillingGadgets3}]{\includegraphics[width=.25\textwidth]{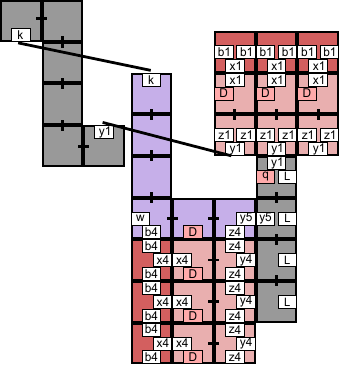}}
    \qquad %
    \subfloat[\label{fig:innerFilling1}]{\includegraphics[width=.25\textwidth]{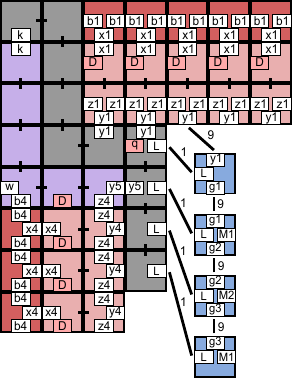}}
    \qquad %
    \subfloat[\label{fig:innerFilling2}]{\includegraphics[width=.25\textwidth]{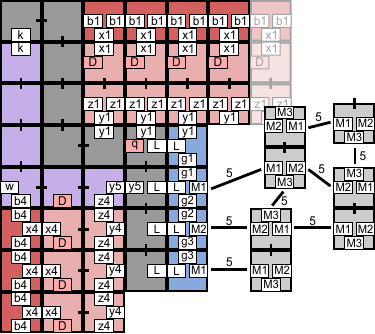}}
    \qquad %
    \caption{The inner filling process. (a) The first filler gadget comes in to fill in the Northernmost-Westernmost corner of HOLLOW[$\Upsilon$]. $k + y1 = 3 + 9 \ge \tau.$ (b) $L=1, g*=9.$ (c) $M1=M2=M3=5.$ First, the 4 tiles in (b) attach one by one to the hollow shape. The filler gadget in (c) can then attach and flood the inside of the shape.}
    \label{fig:innerFilling}
\end{figure}

\section{Universal Shape Replication}\label{sec:universalShapeReplication}
In this section we formally state the results, including the class of shapes which can be replicated by the replication gadgets discussed in the previous section. The input shape must be sufficiently large for the replication gadgets to assemble in the intended manner, so we define the \emph{feature size} of a shape as follows (we use the same definition as~\cite{Abel:2010:SRT:1873601.1873686}): for two points $a,b$ in the shape, let $d(a,b) = \max(|a_x - b_x|, |a_y - b_y|)$. Then the feature size is the minimum $d(a,b)$ such that $a,b$ are on two non-adjacent edges of the shape.

\begin{theorem}\label{thm:main}
There exists a sleek universal shape replicator $\Gamma = (\sigma, 10)$ for genus-$0$ (hole-free) shapes with feature size of at least $9$.
\end{theorem}

\begin{proof}
The proof follows by constructing an unbounded shape replication system $\Gamma = ( \sigma, 10 )$.
Consider an initial assembly state $\sigma$ consisting of infinite counts of the tiles which construct the replication gadgets shown in Section~\ref{sec:replicationGadgets}.
For any shape $X$ of genus-$0$ and minimum feature size $9$, consider an assembly $\Upsilon$ of shape $X$ with the properties discussed at the beginning of Section~\ref{sec:replicationGadgets} (i.e., the exposed glues on $\Upsilon$ are $N,E,S,$ and $W$ based on edge orientation, and there is a special glue in the Northernmost-Westernmost position of $\Upsilon$ which exposes glues $C1$ to the North and $C2$ to the West).
Then we prove inductively that $\Gamma' = ( \sigma\bigcup\Upsilon, 10 )$ has the following property: for any $n\in\mathbb{N}$ and producible state $S$ of $\Gamma'$, there exists a producible state $S'$ containing at least $n$ terminal assemblies of shape $X$ such that $S\rightarrow S'$.

For the base, note that the $\Gamma'$ follows the process described in Section~\ref{sec:replicationGadgets}.
When the FRAME detaches from HOLLOW[$\Upsilon$], HOLLOW[$\Upsilon$] is filled to generate $1$ terminal assembly of shape $X$.
Then, for any producible state $S$, $\sigma\bigcup\Upsilon\rightarrow S$ by definition, so $S$ contains at least $1$ terminal of shape $X$ or $S\rightarrow \hat{S}$ such that $\hat{S}$ contains at least $1$ terminal of shape $X$.

Now, consider any producible state $A$ such that $A$ has at least $n\in\mathbb{N}$ terminal assemblies of shape $X$.
For $A$ to have at least one such terminal, HOLLOW[$\Upsilon$] must have detached from the FRAME.
When this occurs, the FRAME is still attached to the inner start gadget, so the inner mold/drill process repeats, generating HOLLOW[$\Upsilon$].
Again, HOLLOW[$\Upsilon$] detaches from the FRAME, and then HOLLOW[$\Upsilon$] is filled, resulting in $1$ more terminal assembly of shape $X$.
Therefore $A \rightarrow A'$ such that $A'$ contains at least $n+1$ terminal assemblies of shape $X$.
Then, since $S$ or $\hat{S}$ contains at least $1$ terminal of shape $X$, $S\rightarrow S'$ such that $S'$ has at least $n$ terminal assemblies of shape $X$ for any $n$. This satisfies the first property of Definition~\ref{def}.

%Consider any producible state $S$ of the system. If $S$ has less than $n$ terminal assemblies of shape $X$, consider the process shown in Section~\ref{sec:replicationGadgets}.
%When the FRAME detaches from HOLLOW[$\Upsilon$], HOLLOW[$\Upsilon$] is filled to create a terminal assembly of shape $X$.
%Additionally, when the FRAME detaches, the mold and drilling process repeats to create another HOLLOW[$\Upsilon$] assembly.
%Since this process repeats indefinitely, let $S'$ be the state in which this process has repeated enough times for $S'$ to contain at least $n$ terminal assemblies of shape $X$.
%Then $S\rightarrow S'$.

Further, any terminal assembly $B$ in the system that does not have shape $X$, $|B| = \mathcal{O}(1)$, satisfying the second property of Definition~\ref{def}.
Note that the only assemblies which are larger than $\mathcal{O}(1)$ size are assemblies related to the input shape, the FRAME, and HOLLOW[$\Upsilon$].
The input shape and the intermediate assemblies produced while creating the FRAME are not terminal since the process repeats once the FRAME is detached.
The FRAME and the intermediate assemblies produced while creating HOLLOW[$\Upsilon$] are not terminal since the process repeats once HOLLOW[$\Upsilon$] detaches.
HOLLOW[$\Upsilon$] fills to create a terminal assembly of shape $X$, so none of the intermediate assemblies are terminal.

Then, $\Gamma = ( \sigma, 10 )$ is a universal shape replicator for genus-$0$ (hole-free) shapes with minimum feature size $9$.
To show that it is a sleek universal shape replicator (i.e. the bonus properties discussed in the definitions apply), note that the input shape considered here has the same properties as described in that definition, and that the input shape, the FRAME, HOLLOW[$\Upsilon$], and the intermediate steps of each of these assemblies, do not grow larger than $\mathcal{O}(|X|)$ in size.
%Consider $\sigma$ with $0$ terminal assemblies of shape $X$ (the input shape $\Upsilon$ is not terminal).
%$\sigma \rightarrow A$ such that $A$ contains one terminal assembly of shape $X$ by the process shown in Section~\ref{sec:replicationGadgets}.
%For $A$ to contain $1$ terminal assembly of shape $X$, the FRAME must have detached from HOLLOW[$\Upsilon$].
%When the FRAME detaches, HOLLOW[$\Upsilon$] is filled in to a terminal assembly of shape $X$, while the FRAME creates another
\end{proof} 

\section{Conclusion}\label{sec:conclusion}

In this work we formally introduced the problem of shape replication in the 2-handed tile assembly model and provided a universal replication system for all genus-0 shapes with at least a constant minimum feature size.  Shape replication has been studied in more powerful self-assembly models such as the staged self-assembly model and the signal tile model.  However, our result constitutes the first example of general shape replication in a passive model of self-assembly where no outside experimenter intervention is required, and where the system monomers are state-less, static pieces that interact based purely on the attraction and repulsion of surface chemistry.

Our work opens up a number of directions for future work.  One direction includes analyzing and improving the \emph{rate} of replication.  Under a reasonable model of replication time, our construction should costitute a \emph{quadratic} replicator$-$ in that after time $t$, $\Omega(t^2)$ copies of the input shape are expected.  Designing a faster replicator is an open question.  In particular, achieving an \emph{exponential} replicator is an important goal.  Further, to properly consider these questions requires a formal modelling of replication rates for this model.

Another direction is to further generalize the class of shapes that can be replicated.  For example, can shapes with holes be replicated?  This is likely difficult, but might be achievable by drilling into the input assembly in a carefully engineered way.  Achieving such general replication would be the first example of shape replication that extends beyond genus-0. This seems to require a shape with more complex glues, which leads to another area of research.

The consideration of variations of the \emph{sleek} requirements may be of interest.  For example, removing the need for any special tiles from the replication system might be achievable.  Or, allowing for more complex input assemblies could allow for high genus replication, as discussed above.  Finally, determining the lowest necessary temperature and glue strengths needed for replication is an open question.  We use temperature value 10 to maximize clarity of the construction and have not attempted to optimize this value.  To help such optimization, we have included a compiled table (Table~\ref{tbl:inequalities}) showing the inequality specifications induced by our construction for each gadget to help guide where a modification to the replication algorithm might reduce the temperature needed.

Finally, extending replication to work in a planar fashion is an open question, as the current construction requires a large assembly to ``pop'' out of an encased frame.  Planar replication systems might provide insight into extending replication into 3D, while maintaining a \emph{spatial} construction.

\begin{table}
\begin{tabular}{ r r }
\hline
\multicolumn{2}{|c|}{Phase 1} \\
\hline
% start block and outer mold
$C + N \geq \tau$ & $O + N \geq \tau$ \\
$T + W \geq \tau$ \\

\hline
\multicolumn{2}{|c|}{Phase 2: Pre-drilling} \\
\hline
% pre-drilling
$K + C2 + B + D \geq \tau$ & $K + C2 \geq \tau$ \\
$B + W + O + D < \tau$ \\

\hline
\multicolumn{2}{|c|}{Phase 2: Drilling} \\
\hline
% drilling
$B + B \geq \tau$ & $Z + X + D \geq \tau$ \\
$B + X \geq \tau$ \\

\hline
\multicolumn{2}{|c|}{Phase 2: Drilling Type-1 Corners} \\
\hline
% drill acute case 1
$B + W + T + D < \tau$ & $Y + T \geq \tau$ \\
$Z + V + F \geq \tau$ & $Y + Z + T + F \geq \tau$ \\
$Z + B + V + Y + F \geq \tau$ & $Z + X + V + Y + F \geq \tau$ \\
$T + Z + B + F \geq \tau$ & $O + N + O + F < \tau$ \\
$B + B \geq \tau$ & $B + U \geq \tau$ \\

\hline
\multicolumn{2}{|c|}{Phase 2: Drilling Type-2 Corners} \\
\hline
% drilling acute case 2
$Z + H + F \geq \tau$ & $H + Y + Z + B + F \geq \tau$ \\
$H + Y + Z + X + F \geq \tau$ & $T + Z + B + F \geq \tau$ \\
$O + N + O + F < \tau$ & $N < \tau$ \\
$A \geq \tau$ & $B + B \geq \tau$ \\

\hline
\multicolumn{2}{|c|}{Phase 2: Drilling Obtuse Corners} \\
\hline
% drilling obtuse
$B + N + O + D < \tau$ & $B + B \geq \tau$ \\
$Z + R + S + D \geq \tau$ & $O + W + B + D < \tau$ \\

\hline
\multicolumn{2}{|c|}{Phase 2: Post-Drilling} \\
\hline
% post drilling
$B + N + C1 + K + D \geq \tau$ & $Y + Y + Q \geq \tau$ \\
$Y + Z + B + N + C1 + K + Q \geq \tau$ & $B + N + C1 + C2 + B + Z + Y + Q + D \geq \tau$ \\
$Y + Z + B + C2 + K + Q \geq \tau$ & $B + N + C1 + K + D + Q \geq \tau$ \\
$C2  + C1 + N \geq \tau$ & $K + B + C2 + Q + D < \tau$ \\

\hline
\multicolumn{2}{|c|}{Phase 3} \\
\hline
% inner start block, outer mold
$C2 + B \geq \tau$ & $o + n \geq \tau$ \\
$t + e \geq \tau$ \\

\hline
\multicolumn{2}{|c|}{Phase 4 - Pre-drilling} \\
\hline
% pre-drilling
$k + w + b + D \geq \tau$ & $b + w + i + D < \tau$ \\
$k + w \ge \tau$ \\

\hline
\multicolumn{2}{|c|}{Phase 4 - Drilling} \\
\hline
% drilling
$b + b \geq \tau$ & $z + x  + D \geq \tau$ \\
$b + w + i + D < \tau$ & $b + x \geq \tau$ \\

\hline
\multicolumn{2}{|c|}{Phase 4 - Post-drilling} \\
\hline
% post drilling
$B + b + k + C2 + D \geq \tau$ & $y + y + Q \geq \tau$ \\
$B + b + k + C2 + D + Q \geq \tau$ & $b + k + w + D + q < \tau$ \\
$B + C2 \geq \tau$ & $B +  C2 + C1 + N + B1 + O + D \geq \tau$ \\
$Z + B \geq \tau$ & $B + W + W + O + D \geq \tau$ \\
$k + y \ge \tau$ \\

\hline
\multicolumn{2}{|c|}{Phase 5} \\
\hline
$L + g* \geq \tau$ & $M+M \geq \tau$ \\
\end{tabular}
\caption{Shown are the inequalities which must be satisfied for the replication gadgets shown in Section~\ref{sec:replicationGadgets} to function in the way required to prove Theorem~\ref{thm:main}. All single glue labels except $A$  must have strength $< \tau$. Unless otherwise stated, in these inequalities, a glue label $G$ represents all glues $G*$ (e.g., G1, G2, etc.).}
\label{tbl:inequalities}
\end{table}

%\section{Precise Count Shape Replication}\label{sec:preciseCountShapeReplication}
%\input{preciseCount.tex}

\newpage

\bibliographystyle{amsplain}
\bibliography{tam}

\providecommand{\bysame}{\leavevmode\hbox to3em{\hrulefill}\thinspace}
\providecommand{\MR}{\relax\ifhmode\unskip\space\fi MR }
% \MRhref is called by the amsart/book/proc definition of \MR.
\providecommand{\MRhref}[2]{%
  \href{http://www.ams.org/mathscinet-getitem?mr=#1}{#2}
}
\providecommand{\href}[2]{#2}
\begin{thebibliography}{10}

\bibitem{Abel:2010:SRT:1873601.1873686}
Zachary Abel, Nadia Benbernou, Mirela Damian, Erik~D. Demaine, Martin~L.
  Demaine, Robin Flatland, Scott~D. Kominers, and Robert Schwelle, \emph{Shape
  replication through self-assembly and rnase enzymes}, Proceedings of the
  Twenty-first Annual ACM-SIAM Symposium on Discrete Algorithms (Philadelphia,
  PA, USA), SODA '10, Society for Industrial and Applied Mathematics, 2010,
  pp.~1045--1064.

\bibitem{2HABTO}
Sarah Cannon, Erik~D. Demaine, Martin~L. Demaine, Sarah Eisenstat, Matthew~J.
  Patitz, Robert Schweller, Scott~M. Summers, and Andrew Winslow, \emph{Two
  hands are better than one (up to constant factors): Self-assembly in the 2ham
  vs.\ atam}, Proceedings of the 30th International Symposium on Theoretical
  Aspects of Computer Science (STACS 2013), LIPIcs, vol.~20, 2013,
  pp.~172--184.

\bibitem{AGKS05g}
Qi~Cheng, Gagan Aggarwal, Michael~H. Goldwasser, Ming-Yang Kao, Robert~T.
  Schweller, and Pablo~Moisset de~Espan\'{e}s, \emph{Complexities for
  generalized models of self-assembly}, SIAM Journal on Computing \textbf{34}
  (2005), 1493--1515.

\bibitem{Doty2013}
David Doty, Lila Kari, and Beno{\^i}t Masson, \emph{Negative interactions in
  irreversible self-assembly}, Algorithmica \textbf{66} (2013), no.~1,
  153--172.

\bibitem{Fochtman2015}
Tyler Fochtman, Jacob Hendricks, Jennifer~E. Padilla, Matthew~J. Patitz, and
  Trent~A. Rogers, \emph{Signal transmission across tile assemblies: 3d static
  tiles simulate active self-assembly by 2d signal-passing tiles}, Natural
  Computing \textbf{14} (2015), no.~2, 251--264.

\bibitem{Hendricks2015}
Jacob Hendricks, Matthew~J. Patitz, and Trent~A. Rogers, \emph{Replication of
  arbitrary hole-free shapes via self-assembly with signal-passing tiles},
  pp.~202--214, Springer International Publishing, Cham, 2015.

\bibitem{keenan14exponential}
Alexandra Keenan, Robert Schweller, and Xingsi Zhong, \emph{Exponential
  replication of patterns in the signal tile assembly model}, Natural Computing
  \textbf{14} (2015), no.~2, 265--278.

\bibitem{PRS2016RMN}
Matthew~J. Patitz, Trent~A. Rogers, Robert Schweller, Scott~M. Summers, and
  Andrew Winslow, \emph{Resiliency to multiple nucleation in temperature-1
  self-assembly}, DNA Computing and Molecular Programming, Springer
  International Publishing, 2016.

\bibitem{rgTAM}
MatthewJ. Patitz, RobertT. Schweller, and ScottM. Summers, \emph{Exact shapes
  and turing universality at temperature 1 with a single negative glue}, DNA
  Computing and Molecular Programming, LNCS, vol. 6937, 2011, pp.~175--189.

\bibitem{REIF20111592}
John~H. Reif, Sudheer Sahu, and Peng Yin, \emph{Complexity of graph
  self-assembly in accretive systems and self-destructible systems},
  Theoretical Computer Science \textbf{412} (2011), no.~17, 1592 -- 1605.

\bibitem{Rothemund01022000}
Paul W.~K. Rothemund, \emph{Using lateral capillary forces to compute by
  self-assembly}, Proceedings of the National Academy of Sciences \textbf{97}
  (2000), no.~3, 984--989.

\bibitem{Roth01}
Paul W.~K. Rothemund, \emph{Theory and experiments in algorithmic
  self-assembly}, Ph.D. thesis, University of Southern California, December
  2001.

\bibitem{SCHUL2012}
Rebecca Schulman, Bernard Yurke, and Erik Winfree, \emph{Robust
  self-replication of combinatorial information via crystal growth and
  scission}, PNAS \textbf{109} (2012), no.~17, 6405--6410.

\bibitem{SS2013FEC}
Robert Schweller and Michael Sherman, \emph{Fuel efficient computation in
  passive self-assembly}, SODA 2013: Proceedings of the 24th Annual ACM-SIAM
  Symposium on Discrete Algorithms, SIAM, 2013, pp.~1513--1525.

\bibitem{MemoryAlpha}
Memory~Alpha Wikia, \emph{Replicator},
  \url{http://memory-alpha.wikia.com/wiki/Replicator}.

\end{thebibliography}

%\appendix
%\input{appendix.tex}

\end{document}